\documentclass[aps,onecolumn,superscriptaddress]{revtex4}
\usepackage{amsmath}
\usepackage[utf8]{inputenc}
\usepackage{pdfpages}
\usepackage{graphicx}
\usepackage{physics}
\usepackage{amsmath}
\usepackage{epstopdf}
\usepackage{amssymb}
\usepackage{mathtools}
\usepackage{amsthm}
\newtheorem{theorem}{Theorem}
\usepackage[toc,page]{appendix}
\usepackage{float}
\usepackage[colorlinks]{hyperref}
\AtBeginDocument{%
  \hypersetup{
    citecolor=blue,
    linkcolor=blue,
    urlcolor=blue}}

\usepackage{natbib}
\usepackage{bibunits}
\defaultbibliographystyle{naturemag}

\usepackage{bm}
\usepackage{mathtools}
\usepackage{dsfont}
\usepackage{braket}
\usepackage{xcolor}
\usepackage{mdframed}
\usepackage{tabularx}

\usepackage[a4paper, top=3cm, bottom=2cm, left=2.5cm, right=2.5cm]{geometry}

\definecolor{lightorange}{rgb}{0.93,0.93,0.93}
\renewcommand{\subsection}[1]{\par\vspace{1.5em}\noindent\textbf{#1.}\ }

\DeclarePairedDelimiter\floor{\lfloor}{\rfloor}
\usepackage[official]{eurosym}

\newcommand{\absnorm}[1]{\left \lVert #1 \right\rVert_1}

\newcommand{\epsVS}{\epsilon_{{}_{VS}}}
\newcommand{\epsBetaPM}{\epsilon_{{}_{\beta PM}}}
\newcommand{\piVS}{\pi_{{}_{VS}}}

\usepackage[acronym]{glossaries}

\newacronym{ISG}{ISG}{Industry Specification Group}
\newacronym{POWF}{POWF}{Physical One-Way Function}
\newacronym{PUF}{PUF}{Physical Unclonable Funcion}
\newacronym{ISO}{ISO}{International Organization for Standardization}
\newacronym{IEC}{IEC}{International Electrotechnical Commission}
\newacronym{PNS}{PNS}{Photon-Number-Splitting}
\newacronym{BB}{BB84}{Brassard-Bennet-1984}
\newacronym{GG}{GG02}{Grosshans-Grangier-2002}
\newacronym{PLOB}{PLOB}{Pirandola-Laurenza-Ottaviani-Bianchi}
\newacronym{PQC}{PQC}{Post-Quantum Cryptography}
\newacronym{RNG}{RNG}{Random Number Generator}
\newacronym{QC}{QC}{Quantum Cryptography}
\newacronym{QKD}{QKD}{Quantum Key Distribution}
\newacronym{BPM}{$\beta$PM}{$\beta$-Partial Matching}
\newacronym{VS}{VS}{Vector in a Subspace}
\newacronym{IT}{IT}{Information Technology}
\newacronym{TOE}{TOE}{Target Of Evaluation}
\newacronym{AP}{AP}{Attack Potential}
\newacronym{CPTP}{CPTP}{Complete Positive Trace Preserving}
\newacronym{CP}{CP}{Complete Positive}
\newacronym{CVQKD}{CV-QKD}{Continuous-Variable Quantum Key Distribution}
\newacronym{DVQKD}{DV-QKD}{Discrete-Variable Quantum Key Distribution}
\newacronym{CAC}{CAC}{Classical Authenticated Channel}
\newacronym{QDL}{QDL}{Quantum Data Locking}
\newacronym{NSM}{NSM}{Noisy-Storage Model}
\newacronym{QCT}{QCT}{Quantum Computational Time-lock}
\newacronym{EPR}{EPR}{Einstein-Podolsky-Rosen}
\newacronym{AES}{AES}{Advanced Encryption Standard}
\newacronym{RSA}{RSA}{Rivest-Shamir-Adleman}
\newacronym{BQSM}{BQSM}{Bounded-Quantum-Storage Model}
\newacronym{HMQCT}{HM-QCT}{Hidden Matching Quantum Computational Time-lock}
\newacronym{QBER}{QBER}{Quantum Bit Error Rate}
\newacronym{SKC}{SKC}{Secret Key Capacity}
\newacronym{SNSPDs}{SNSPDs}{Superconducting Nanowire Single-Photon Detectors}
\newacronym{DIQKD}{DI-QKD}{Device-Independent Quantum Key Distribution}
\newacronym{DDQKD}{DD-QKD}{Device-Dependent Quantum Key Distribution}
\newacronym{MDIQKD}{MDI-QKD}{Measurement-Device-Independent Quantum Key Distribution}
\newacronym{TFQKD}{TF-QKD}{Twin-Field Quantum Key Distribution}
\newacronym{POVM}{POVM}{Positive Operator-Valued Measure}
\newacronym{NA}{NA}{Numerical Aperture}
\newacronym{MMF}{MMF}{multimode fiber}
\newacronym{TM}{TM}{Transmission Matrix}
\newacronym{SLM}{SLM}{Spatial Light Modulator}
\newacronym{MEMS}{MEMS}{Micro Electro-Mechanical System}
\newacronym{DMD}{DMD}{Digital Micromirror Device}
\newacronym{EMCCD}{EMCCD}{Electron-multiplying  Charge-Coupled Device}
\newacronym{VBS}{VBS}{Variable Beam Splitter}
\newacronym{BS}{BS}{Beam Splitter}
\newacronym{PBS}{PBS}{Polarization Beam Splitter}
\newacronym{EB}{EB}{Entangled-Based}
\newacronym{GMCS}{GMCS}{Gaussian-Modulated Coherent State}
\newacronym{RMT}{RMT}{Random Matrix Theory}
\newacronym{QSA}{QSA}{Quantum-Secure Authentication}
\newacronym{DMPK}{DMPK}{Dorokhov-Mello-Pereira-Kumar}
\newacronym{LO}{LO}{Local Oscillator}
\newacronym{PLS}{PLS}{Physical Layer Security}
\newacronym{SVD}{SVD}{Singular Value Decomposition}
\newacronym{CV}{CV}{Continuous Variable}
\newacronym{PM}{PM}{Prepare and Measure}

\newtheorem{corollary}{Corollary}
\newtheorem{lemma}{Lemma}

\usepackage[T1]{fontenc}
\usepackage[utf8]{inputenc}
\usepackage{lmodern}

\makeatletter
\renewcommand{\subsubsection}[1]{%
  \par\vspace{0.8em}%
  \noindent\textit{#1.}\quad
}
\makeatother

\begin{document}

\begin{bibunit}[naturemag]
\title{Reconfigurable Optical Platform for One-way Quantum Communication Complexity}

\author{Francesco Mazzoncini} 
\affiliation{T\'el\'ecom Paris-LTCI, Institut Polytechnique de Paris, 19 Place Marguerite Perey, 91120 Palaiseau, France}
\affiliation{Orange Innovation, Orange Gardens, 44 Avenue de la République, 92326 Châtillon, France}

\author{Hugo Defienne}
\affiliation{Sorbonne Université, CNRS, Institut des NanoSciences de Paris, INSP, Paris, France}

\author{Romain Alléaume}
\affiliation{T\'el\'ecom Paris-LTCI, Institut Polytechnique de Paris, 19 Place Marguerite Perey, 91120 Palaiseau, France}
\author{Sylvain Gigan}
\affiliation{Laboratoire Kastler Brossel, ENS-Université PSL, CNRS, Sorbonne Université, Collège de France, 24 rue Lhomond, Paris, France}
\date{\today}

\begin{abstract}

Demonstrating a practical quantum advantage remains a central goal in quantum information science. While quantum computational supremacy is still technologically demanding, communication complexity offers a promising route to showcase quantum advantage with current photonic platforms. Here we introduce a reconfigurable optical platform for one-way quantum communication complexity based on multimode fibers and wavefront shaping. We experimentally validate it by implementing a genuine one-way quantum communication complexity problem for which an exponential quantum--classical communication separation is known. Complementary  numerical simulations show that the same reconfigurable decoding architecture can support more general one-way communication tasks with comparable performance, while also offering a route to higher-dimensional implementations without increasing hardware complexity.
Together, these results establish multimode-fiber wavefront shaping as a versatile hardware platform for one-way quantum communication complexity and provide a concrete roadmap toward more demanding protocols, where stronger quantum--classical separations could enable practical demonstrations of quantum advantage.

\end{abstract}

\maketitle
\section{Introduction}

A central objective in quantum information science is to identify and experimentally demonstrate tasks where quantum resources provide a clear advantage over their classical counterparts, whether in terms of computational power \cite{harrowQuantumComputationalSupremacy2017, gillQuantumComputingTaxonomy2022}, security \cite{gisinQuantumCryptography2002, pirandolaAdvancesQuantumCryptography2020}, or communication efficiency \cite{dewolfQuantumCommunicationComplexity2002, buhrmanNonlocalityCommunicationComplexity2010}. While the realization of quantum advantage in general-purpose quantum computation remains a formidable challenge, the field of communication complexity offers a promising avenue for showcasing quantum advantage with currently available photonic technologies.

Communication complexity \cite{yaoProbabilisticComputationsUnified1977} studies the minimum amount of information that must be exchanged between parties to jointly solve a distributed computational task. Over the years, communication complexity has been extensively investigated across multiple layout models, which specify the allowed interaction between the players.
In the two-way (interactive) communication model, players can exchange messages interactively until one of them produces the answer. In the one-way communication model, only Alice sends a message to Bob, who must produce the answer based on the message he received and his input. In the simultaneous message passing (SMP) model, both Alice and Bob send one message to a third participant (the referee), who then produces an answer based on these two messages.  Quantum protocols  have been shown to offer  exponential advantage in communication cost across the standard layouts, with the comparison with classical protocols made within the same model in most cases. For example, in the interactive model we have Raz’s  pioneering two-way communication problem \cite{razExponentialSeparationQuantum1999}, in the one-way model the Hidden Matching problem \cite{bar-yossefExponentialSeparationQuantum2004}, and in the SMP model  the Quantum Fingerprinting problem \cite{buhrmanQuantumFingerprinting2001}. Beyond these, there is a task, known as \gls{VS} \cite{regevQuantumOnewayCommunication2011}, where a one-way quantum protocol still achieves an exponential advantage even against the stronger classical two-way interactive model.   However, implementing these protocols in practice is often hindered by the need for complex, large-scale interferometric networks or highly entangled multi-qubit states, which are difficult to realize with current experimental platforms.

Recent theoretical advances have proposed practical mappings of quantum communication protocols onto optical systems using coherent states, linear optics, and single-photon detection \cite{arrazolaQuantumCommunicationCoherent2014, arrazolaQuantumFingerprintingCoherent2014}. Proof-of-principle demonstrations followed, but they relied on  detection setups tailored to narrow families of transformations rather than fully reconfigurable, high-dimensional linear networks. In particular, in SMP quantum fingerprinting the referee uses a fixed phase interferometer where Alice's and Bob's  optical pulses interfere through a balanced beam splitter \cite{xuExperimentalQuantumFingerprinting2015a, zhongEfficientExperimentalQuantum2021};  while in the one-way Hidden Matching experiment, Bob's cascaded delay-line interferometers implement a fixed  set of transformations over the time-bin register \cite{kumarExperimentalDemonstrationQuantum2019}. More precisely, Ref.~\cite{kumarExperimentalDemonstrationQuantum2019} discusses the general quantum protocol for Hidden Matching, but experimentally implements a related variant, known as Sampling Matching, which is not strictly a communication complexity task since Bob does not receive an input. However, when aiming for an experimental advantage of a one-way quantum protocol even against the stronger classical two-way model, a far more flexible  detection platform is needed---capable of implementing programmable random linear transformations on high-dimensional quantum states \cite{regevQuantumOnewayCommunication2011}.

In this work, we introduce and experimentally demonstrate  such a reconfigurable optical platform for one-way quantum communication complexity problems, leveraging the rich mode-mixing properties of multimode fibers (MMFs) and wavefront shaping techniques \cite{popoffMeasuringTransmissionMatrix2010, rotterLightFieldsComplex2017}. Our approach enables the implementation of high-dimensional, programmable linear optical networks without the need for cascaded interferometers or extensive optical hardware \cite{leedumrongwatthanakunProgrammableLinearQuantum2020,cavaillesHighfidelityLargescaleReconfigurable2022, makowskiLargeReconfigurableQuantum2024}.  By encoding information in the spatial degrees of freedom of light and dynamically controlling the input wavefront, we realize a flexible and low-loss platform capable of addressing a wide range of quantum communication complexity problems.

We focus on two benchmark problems with well-understood quantum–classical gaps across communication models: the already mentioned \gls{VS} problem and the \gls{BPM} problem \cite{gavinskyExponentialSeparationsOneway2007, mazzonciniHybridQuantumCryptography2025a} (a variant of Hidden Matching). The \gls{VS} problem admits a one-way quantum protocol using only $\mathcal{O}(\log n)$ qubits, while any classical two-way protocol requires polynomial communication—provably $\Omega(n^{1/3})$  and conjectured $\Omega(\sqrt{n})$  —thus yielding an exponential separation even against the stronger classical model \cite{regevQuantumOnewayCommunication2011}. By contrast, \gls{BPM} exhibits an exponential separation within the one-way model itself: a quantum $\mathcal{O}(\log n)$ qubits  protocol versus $\Omega(\sqrt{n})$  bits for any classical one-way protocol \cite{gavinskyExponentialSeparationsOneway2007}. However, from an implementation standpoint, \gls{VS} requires full  amplitude-and-phase preparation of Alice's quantum state, whereas \gls{BPM} is phase-only and therefore compatible with our current setup.

Accordingly, we experimentally implement \gls{BPM} on our reconfigurable spatial-mode platform, and we use the best-known one-way classical protocol for \gls{VS} as a stringent benchmark on the transmitted information. To our knowledge, no explicit classical two-way protocol is currently known that would provide a stronger benchmark for this task.
Complementary numerical studies (Methods) show that our programmable decoding stage can realize the \gls{VS} measurement with comparable performance: our current hardware limitation lies in Alice’s amplitude control, not in Bob’s reconfigurability. In particular, our results suggest that increasing the number of supported modes and mitigating technical noise could bring the platform into a regime where a one-way quantum implementation outperforms the best-known classical one-way bounds for \gls{VS}. 

Extending this analysis to the more general two-way classical model---either by tightening existing lower bounds (proving the $\Omega(\sqrt{n})$ conjecture with explicit prefactors) or by developing explicit interactive protocols and simulating their performance---would provide a broader benchmark and refine the experimental targets (mode count, loss budget, noise) for demonstrating an exponential one-way quantum advantage even against classical two-way strategies.

\section{Results}
\subsection{Communication-complexity benchmarks}
\label{sec:QCC}
Communication complexity is a central model of computation, first defined by Yao \cite{yaoProbabilisticComputationsUnified1977}, where there are two players, Alice and Bob, who each receive an input: Alice gets $x$ from a set $\mathcal{X}$ and Bob gets $y$ from a set $\mathcal{Y}$. Their goal is to use the allowed type of communication  (either classical or quantum) to compute with high probability the value of $f(x,y)$, where $f$ is a function (or relation) defining the computational problem that the players have to solve.

\subsubsection{\texorpdfstring{$\beta$}--Partial Matching Problem}
\label{subsec:betapm}
In this section we shall present the quantum communication complexity problem, called \gls{BPM} problem, for which  $\Omega(\sqrt{n})$ bits of classical communication from Alice to Bob are required, with $n$ the length of input $x$.

Suppose Alice has an $ n $-bit string $ x $, while Bob has a sequence $ M $ consisting of $ \beta n $ disjoint pairs $ (i_1, j_1), (i_2, j_2), \ldots, (i_{\beta n}, j_{\beta n}) \in [n]^2 $, for some parameter $ \beta \in (0, 1/2] $. This sequence $ M $ can be interpreted as a partial matching on the graph whose vertices are the $ n $ bits $ x_1, \ldots, x_n $. We refer to this as  $ \beta $-matching.
Together, $ x $ and $ M $ generate a $ \beta n $-bit string $ z $ defined by the parities of the $ \beta n $ edges:

\begin{equation}
    z = z(x, M) = x_{i_1} \oplus x_{j_1}, x_{i_2} \oplus x_{j_2}, \ldots, x_{i_{\beta n}} \oplus x_{j_{\beta n}}.
\end{equation}

 Let $ x \in \{0, 1\}^n $ be uniformly distributed, and let $ M $ be uniformly chosen from the set of all $ \beta $-matchings. It is important to note that for any fixed matching $ M $, a uniform distribution on $ x $ results in a uniform distribution on $ z $. Therefore, Bob, who knows $ M $ but not $ x $, has no information about $ z $: from his viewpoint, it appears uniformly distributed. Now, suppose Bob also has access to a string \( \omega \) which is either equal to \( z \) or \( \bar{z} \), and he wants to determine which of the two cases is true.

 \begin{theorem}[\gls{BPM} best-known classical protocol \cite{mazzonciniHybridQuantumCryptography2025a}]
\label{theo:betapm_communication}
Let $d$ be an integer. An explicit one-way classical protocol exists with communication cost $d$ bits which solves the $n$-dimensional $\beta PM$ protocol with error probability, for any input, at most
\begin{equation}
 \epsBetaPM(d)=  \sum_{k=0}^{d} \frac{\binom{2 \beta n}{k} \binom{n-2 \beta n}{d-k}}{2\binom{n}{d}} e^{-\frac{k(k-1)}{4\beta n}}\;,
 \label{eq:eps_betaPM}
\end{equation}
where we use the convention $\binom{a}{b}=0$ whenever $b>a$.
\end{theorem}

On the other hand, using  quantum resources, the above task can be solved with a fixed error probability by transmitting only   $\mathcal{O}(\log n)$ qubits \cite{gavinskyExponentialSeparationsOneway2007}.
Alice prepares a quantum state encoding her bit string as phase information in a uniform superposition 
\begin{equation}
    \ket{\psi_x} = \frac{1}{\sqrt{n}} \sum_{i=1}^n (-1)^{x_i} \ket{i}
    \label{eq:single_phot_psi}
\end{equation}
and sends it to Bob. His measurement can be represented as a linear operator. This operator depends on the input $y = (M, \omega)$, which includes the matching $M$ and the binary vector $\omega$. After applying the operator, the measurement outcome is determined by projecting onto one of two orthogonal subspaces, corresponding to $b=0$ or $b=1$.
We refer to the Supplementary Material for a full description of Bob's linear operator.

\subsubsection{Vector in a Subspace Problem}
The \gls{VS} problem is another communication problem where one party Alice receives a unit vector $x \in \mathbb{R}^n$ and Bob receives a subspace $H\subseteq \mathbb{R}^n$ of dimension $\floor{n/2}$  such that either  $x \in H$ or $x \in H^{\perp}$. Their goal is to answer 0 if $x$ is in $H$ and 1 if $x$ is in $H^{\perp}$. In particular, we will focus on the case where both the unit vector $x$ and the subspace $H$ are picked from a uniform distribution, with the promise that $x$ is either in $H$ or $H^{\perp}$.

In the study by Regev et al.\ \cite{regevQuantumOnewayCommunication2011}, the authors established a lower bound of $\Omega(n^{1/3})$ for the two-way classical communication complexity. However, the lower bound presented in their work is not considered to be definitive. They propose that a lower bound of $\Omega(\sqrt{n})$ is likely to be more accurate. In support of this conjecture, a restricted version of the problem was examined in the study by Grupel et al.\ \cite{grupelSamplingSphereMutually2017}, where they established an optimal lower bound of $\Omega(\sqrt{n})$ bits. 

Finally, a classical protocol was sketched \cite{razExponentialSeparationQuantum1999}, which solves the $n$-dimensional \gls{VS} problem protocol with one-way communication complexity of $\mathcal{O}(\sqrt{n})$. Our original contribution has been to carefully derive the prefactors of this sketched protocol. Similar to the \gls{BPM} problem, we only introduce the key theorem here that establishes an upper boundary for the error probability. The proof of this theorem can be found in the Supplementary Material.
\begin{theorem}[VS best-known classical protocol]
\label{theo:VS_communication}
Let $n\ge 4$ be even and let $d\ge 1$ be an integer. An explicit one-way public-coin protocol $\piVS$ exists with communication cost $CC(\piVS)=d$ which solves the $n$-dimensional \gls{VS} problem with error probability at most
\begin{equation}
\label{eq:epsilon_VS}
\epsVS(d,n)
=
\inf_{\begin{array}{c}
0<t<\sqrt{d/\pi},\,
u>0
\end{array}}
\left[
e^{-t^2/2}
+
2e^{-u^2/2}
+
e^{-x_{t,u}}
\right],
\end{equation}
where $x_{t,u}$ is defined as follows. If the set of nonnegative numbers $x$ such that
\begin{equation}
\label{eq:def_xtu}
\frac{n-1+2\sqrt{(n-1)x}+2x}{n}
\le
\frac{\left(\sqrt{d/\pi}-t\right)^2}{\frac{n}{n-1}+\frac{2un}{\sqrt{n-1}}}
\end{equation}
is nonempty, then $x_{t,u}$ is its largest element. Otherwise, we set $x_{t,u}=0$.
\end{theorem}

On the other hand, a one-way quantum protocol is known that only requires $\mathcal{O}(\log n)$ qubits \cite{regevQuantumOnewayCommunication2011}  to solve the VS problem with a fixed error probability.
The protocol involves Alice creating a quantum state $\ket{\psi_x}= \sum_{i=1}^{n} x_i \ket{i}$ 
and transmitting it to Bob, who then executes a projective measurement onto the subspaces $H$ and $H^{\perp}$. In this protocol, Alice's encoding is based on the relative amplitudes, contrasting with the \gls{BPM} quantum protocol where the information is encoded in the relative phases.

\subsubsection{Comparison between  the two problems}
As we conclude this section, one important question remains open: which of the two protocols is harder to implement in the classical model? 
For a fixed error probability threshold (e.g., 0.2), a meaningful comparison can be drawn between the transmitted bits for the best classical protocol with the number of transmitted qubits in a one-way quantum protocol. 
In particular, we allow Alice to send $m$ identical copies of the same $n$-dimensional quantum state, thereby reducing the error probability.
Overall, given $m$ copies sent, the amount of qubits transmitted is simply upper bounded by $m\lceil\log(n)\rceil$.

 \begin{figure}[htb!]
\centering
\vspace{1pt}
\includegraphics[width=0.7\textwidth]{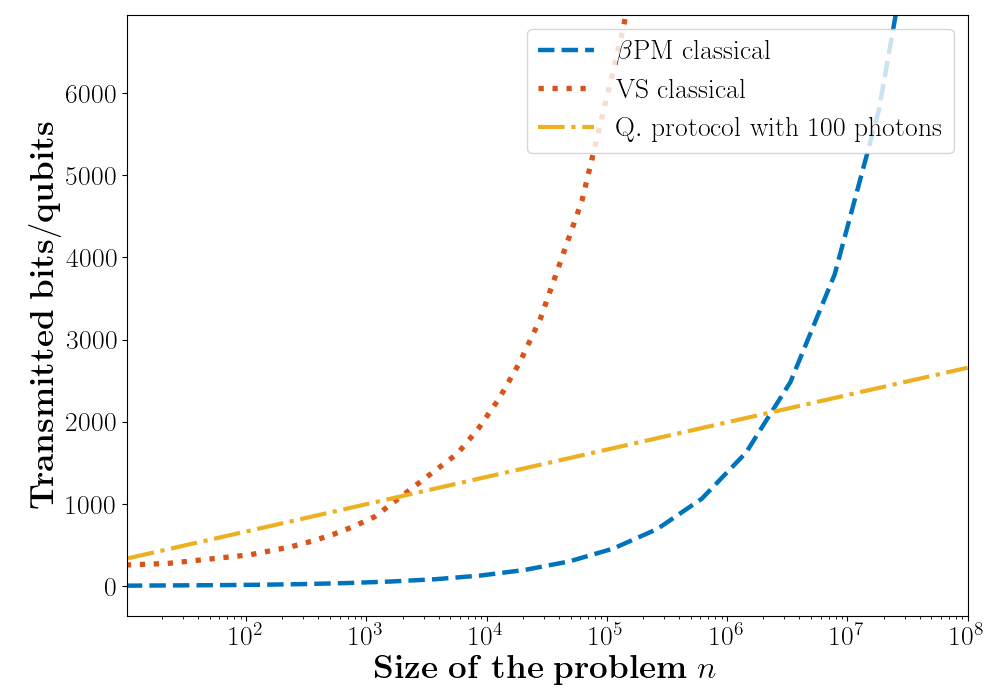}
\caption{Comparison between the required number of  transmitted bits/qubits for classical and quantum upper bounds for an error probability of $\epsilon = 0.2$. }
\label{fig:comparison_VS_HM}
\end{figure}

As observed in Figure \ref{fig:comparison_VS_HM},  at a fixed size of the problem $n$, the  \gls{VS} classical protocol requires to  send many more bits compared to the \gls{BPM} problem.  The plot also illustrates how one could implement the quantum protocol for the \gls{VS} problem  with 100 copies of the same quantum state and still outperform its classical upper bound even for relatively modest values of $n$ (e.g., a few thousand). This redundancy in quantum state transmission not only enhances the protocol's practicality but also provides crucial compensation for potential transmission losses while increasing resilience to environmental noise.

\subsection{Reconfigurable spatial-mode platform}
As illustrated in Figure \ref{fig:setup_true}, our platform to test one-way communication complexity problems is conceptually divided into two parts. 
 
 \begin{figure}[htb!]
         \centering
         \includegraphics[width=0.9\textwidth]{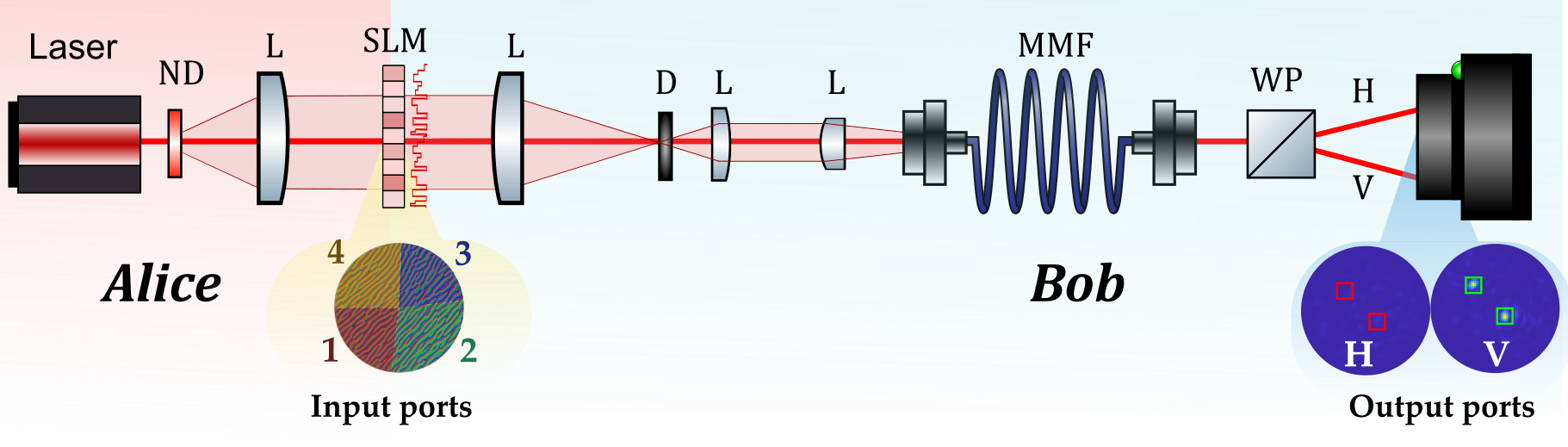}
\caption[Setup configuration to solve a one-way boolean quantum communication complexity problem.]{Setup configuration to solve a one-way boolean quantum communication complexity problem. Alice's encoding is simulated by adding an additional mask to the \gls{SLM}, which has been carefully partitioned in $n$ ports ($n=4 $ in this illustration). Bob uses the same \gls{SLM} and  \gls{MMF} to build a linear operator. He then divides the output in the camera into two subspaces (light H-polarized and V-polarized), where for each subspace he defines $k/2$ macro pixels as output ports ($k=4$ in this illustration). 
If Bob detects a greater amount of light in the red detection modes, he responds with $b=0$; if not, he responds with $b=1$.(L: lens, ND: neutral-density filter, D: Iris diaphragm, WP: Wollaston prism.)}
\label{fig:setup_true}
 \end{figure}

\begin{itemize}
\item \textbf{Alice's setup:}
Alice employs a superluminescent diode at $810 \pm 20$ nm to generate coherent states with a controllable mean photon number. A 1-nm bandpass filter is added on the source path so that spectral dispersion is negligible when light propagates through the \gls{MMF}. This is justified by the use of a short graded-index fiber, for which the spectral bandwidth $\delta\omega$ associated with output-speckle decorrelation is much larger than the spectral width of the filtered source. Alice also has shared access to a reflective liquid-crystal phase-only \gls{SLM} (Hamamatsu, X10468-02), which is partitioned into $n$ input ports, as described in the Supplementary Material.

\item \textbf{Bob's setup:}
Bob's apparatus includes the shared \gls{SLM} and a graded-index \gls{MMF} supporting approximately 400 propagation modes at 810 nm (Thorlabs, GIF50C, length $55.3 \pm 0.1$ cm, core diameter $50 \pm 2.5$ $\mu$m, \gls{NA} $0.200 \pm 0.015$). Since the \gls{SLM} controls only the H-polarized input field, we effectively address approximately 200 modes. The \gls{MMF} mixes spatial and polarization modes, and a Wollaston prism is used to separate the two orthogonal polarizations, H and V, onto distinct regions of the \gls{EMCCD} camera (Andor iXon3 860). This binary division of the detection plane is specifically designed for boolean one-way communication complexity problems.

\end{itemize}

It's worth noting that in a practical implementation of a one-way quantum communication protocol, two \gls{SLM}s would be necessary: one for Alice to encode her quantum states and another for Bob to perform the correct measurements. However, for the sake of simplicity and to minimize costs, our setup employed just a single \gls{SLM}, where Alice’s input state is encoded by adding additional phase masks to each port of the SLM, drawing inspiration from a strategy implemented in \cite{matthesOpticalComplexMedia2019}.

 \subsubsection{Alice's encoding}
 First, as previously mentioned, the components of the input quantum states are encoded directly on the \gls{SLM}.
 Unfortunately,  since the \gls{SLM} can only encode phase masks, we are not able to directly simulate the input state used to solve the \gls{VS} quantum one-way communication complexity protocol. Consequently, our experiment will be limited to the \gls{BPM}  protocol, which requires only phase encoding. However, at a fixed size of the problem, the \gls{VS} has a best-known classical protocol which is much more demanding in terms of communication cost with respect to the \gls{BPM}, making it a much better candidate to show a quantum advantage. To address the limitations imposed by the \gls{SLM}'s encoding capabilities,  we perform a numerical analysis in the Methods section to demonstrate that even when tackling the more general \gls{VS} problem, our programmable linear network can achieve performance levels comparable to those achievable in the more practically implementable \gls{BPM} quantum protocol.

However, motivated by the difficulty in implementing quantum protocols based on single-photon encoding, an alternative mapping of quantum communication complexity protocol was introduced \cite{arrazolaQuantumCommunicationCoherent2014}, where, instead of the single-photon state in Eq.\ \eqref{eq:single_phot_psi}, one simply implements a sequence of coherent states over $n$ optical modes.  Specifically, Alice's quantum state is given by 
\begin{equation}
    \ket{\alpha, x} =\bigotimes_{p=1}^{n} \left |\frac{\alpha}{\sqrt{n}}(-1)^{x_p}\right \rangle_p\;,
    \label{eq:alice_coherent}
\end{equation}
where $x \in \{0,1\}^n$ is Alice's input and $\alpha$ is a complex amplitude.
Specifically, the photon number distribution in each mode is the same as if one performed the measurement on multiple copies of the single-photon state from Eq.\ \eqref{eq:single_phot_psi}, with the number of copies following a Poisson distribution of mean $|\alpha|^2$. Because of the Holevo bound~\cite{holevoBoundsQuantityInformation1973}, any quantum state in an $n$-dimensional Hilbert space can carry at most $\log(n)$ bits of information. For this reason, we use the average transmitted information, given by $|\alpha|^2 \log(n)$, as a meaningful metric for the information that can be extracted from the coherent state in Eq.\ \eqref{eq:alice_coherent}.

In our experiment, we exploit this relationship by precisely controlling the average photon number $|\alpha|^2$ through modulation of the superluminescent diode's output intensity and adjustment of the \gls{EMCCD} camera's exposure time. This allows us to directly set the average transmitted information per round, providing a meaningful and fair metric for comparing the efficiency of our quantum protocol to classical benchmarks.

\subsubsection{Bob's decoding}
We are finally ready to explain how to harness the transmission matrix information to build a reconfigurable detection system. The general concept relies on digital phase conjugation \cite{fisherOpticalPhaseConjugation1983, popoffMeasuringTransmissionMatrix2010}, where each $p$-th input port of the  \gls{SLM} is encoded such that  optical field at the input of the \gls{MMF} has the form
\begin{equation}
    E_{in}^{(p)} = \textbf{T}^{(p)\dagger} \bm{\mathcal{L}}^{(p)}\;,
    \label{eq:lin_op}
\end{equation}
where $\bm{\mathcal{L}}^{(p)}$ is the $p$-th column of the linear operator, as illustrated in Figure \ref{fig:lin_op_pizza}.
In particular, since the \gls{SLM} is on the Fourier plane of the input face of the \gls{MMF},
the phase mask displayed on each port of the \gls{SLM} corresponds to the angular component of the Fourier transform of the input field in Eq.\  \eqref{eq:lin_op}.
Therefore, the combination of the light modulation from the \gls{SLM} and the mode mixing of the \gls{MMF} 
 results in an observed linear operator $\bm{\mathcal{L}}_{\text{obs}}$, where the column  encoded by the $p$-th input port has the form
\begin{equation}
    \bm{\mathcal{L}}_{\text{obs}}^{(p)} = \textbf{T}^{(p)}\textbf{T}^{(p)\dagger} \bm{\mathcal{L}}^{(p)}\;.
    \label{eq:L_operator_obs}
\end{equation}
 \begin{figure}[htb!]
         \centering
         \includegraphics[width=0.9\textwidth]{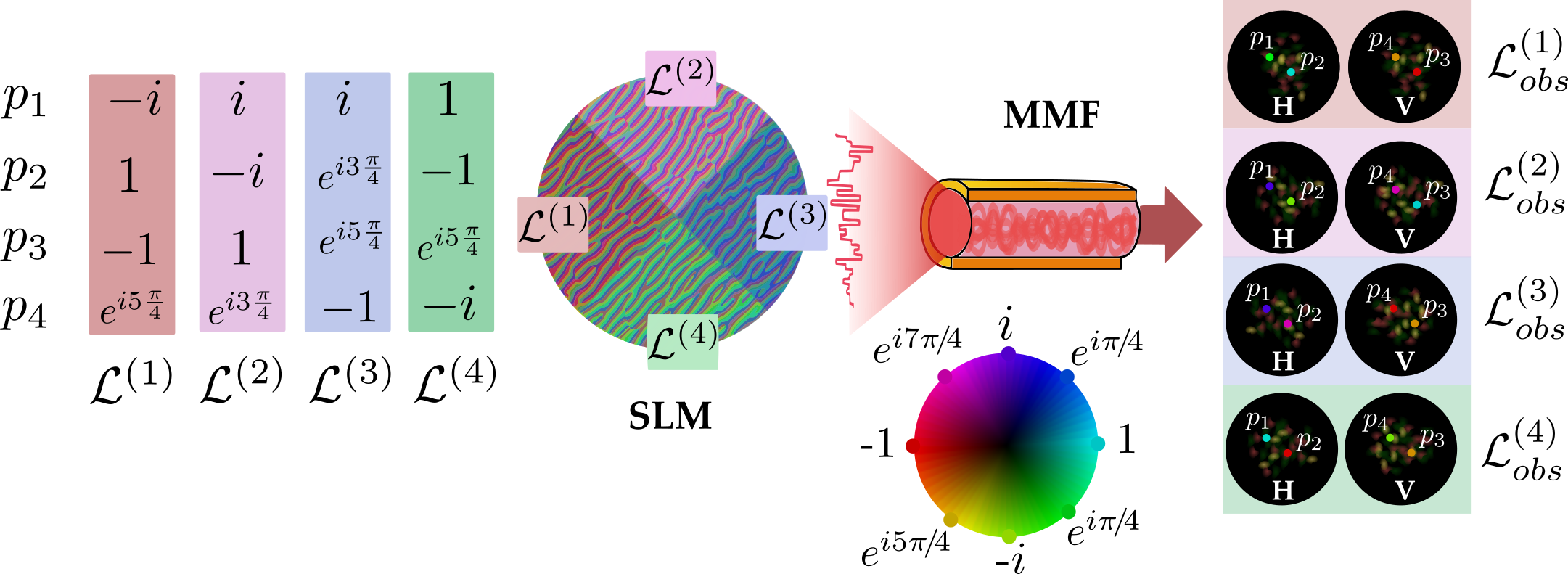}
\caption[Pictorial representation of the construction of a $4\times 4$ linear operator.]{Pictorial representation of the construction of a $4\times 4$ linear operator by Bob. Each port of the \gls{SLM} encodes a different column of the linear operator, controlling both the amplitude and the phase of the corresponding output state. The final output detected is a coherent sum of the contribution of each individual port.}
\label{fig:lin_op_pizza}
 \end{figure}
 
Bob can then program his reconfigurable optical detection system presented in Figure \ref{fig:setup_true} according to his classical input $y$ for the \gls{BPM} problem.
The output state impinging the camera after encoding on the \gls{SLM} Bob's linear operator $\bm{\mathcal{L}}_{\beta \text{PM}}(y)$, described in the Supplementary Material, has hence the form

\begin{align}
         \ket{\alpha, x,y} &= \bigotimes_{j=1}^{k}  \left |\alpha \sqrt{\frac{\eta}{n}}\sum_{p=1}^n (-1)^{x_p}\left(\textbf{T}^{(p)}\textbf{T}^{(p)\dagger} \bm{\mathcal{L}}^{(p)}_{\beta \text{PM}}(y)\right)_j\right \rangle_j \nonumber \\
         &=  \bigotimes_{j=1}^{k}\left |\alpha^{(x,y)}_j\right \rangle_j\;,
         \label{eq:coherent_out_BPM}
\end{align}

\noindent
 where we defined $\alpha^{(x,y)}_j\coloneqq \alpha \sqrt{\frac{\eta}{n}}\sum_{p=1}^n (-1)^{x_p}\left(\textbf{T}^{(p)}\textbf{T}^{(p)\dagger} \bm{\mathcal{L}}^{(p)}_{\beta \text{PM}}(y)\right)_j$ the new complex amplitudes for the  $j$-th output mode and $\eta$  the overall efficiency of optical setup.

Let $D_H$ ($D_V$) be the random variable corresponding to the number of clicks in the  $k/2$ output modes of the camera corresponding to the $H$ ($V$) region. Since the output state  in Eq.\ \eqref{eq:coherent_out_BPM} is still a tensor product of $n$ coherent states, these two variables are distributed according to Poissonian distributions  $D_H \sim \text{Poisson}\left( \mu_H =\sum_{j\in H} \left|\alpha^{(x,y)}_j\right|^2\right)$ and $D_V \sim \text{Poisson}\left(\mu_V =\sum_{j\in V} \left|\alpha^{(x,y)}_j\right|^2\right)$. Considering the bit  $a$ being the right answer, then the error probability $\epsilon$ of the protocol is
\begin{equation}
    \epsilon = 
    \begin{cases} 
        \text{Pr}(D_H < D_V) + \frac{1}{2} \text{Pr}(D_H = D_V) & \text{if } a = 0\;, \\
        \text{Pr}(D_H > D_V) + \frac{1}{2} \text{Pr}(D_H = D_V) & \text{if } a = 1\;.
    \end{cases}
\end{equation}


\subsubsection{Setup efficiency} Before presenting the experimental results, it is important to highlight the main sources of loss in our optical system, as summarized in Table~\ref{table:experiment-parameters}. 
 In particular, the efficiency of the optical pathway is limited by the reflectivity of the SLM and losses at the injection of the fiber. 
This is because the guiding core of the \gls{MMF} acts as a spatial filter (both spatially and in numerical aperture---i.e.\ angularly), and some of the high-frequency components of the \gls{SLM} patterns are filtered. 
The overall efficiency of the optical setup  in Eq.\ \eqref{eq:coherent_out_BPM} is therefore $\eta \approx 50\%$. 

\begin{table}[h] \centering 
\centering
\begin{tabular}{|l|c|}
\hline

\textbf{Parameter} & \textbf{Value} \\

\hline
SLM reflectivity: \,\, $\bm{\eta_{\textbf{SLM}}}$ & 90\% \\
\hline
MMF transmission: \, \, $\bm{\eta_{\textbf{MMF}}}$ & 80\% \\
\hline
Detector quantum efficiency: \,\, $\bm{\eta_{\textbf{det}}}$ & 75\% \\
\hline
Detection ports efficiency: \,\, $\bm{\eta_{\textbf{ports}}}$ & 8\% \\
\hline
\end{tabular}
\caption{Summary of the main sources of loss in the optical setup and their estimated efficiencies. The SLM efficiency ($\eta_{\text{SLM}}$) is determined by its reflectivity, the MMF efficiency ($\eta_{\text{MMF}}$) by the transmission through the fiber including injection losses, and the detection efficiency ($\eta_{\text{det}}$) by the quantum efficiency of the EMCCD camera. The port efficiency ($\eta_{\text{ports}}$) corresponds to the fraction of light collected by the selected detection ports.}\label{table:experiment-parameters} 
\end{table}

However, the main source of loss in the camera arises from utilizing only a limited number of macro pixels as detection ports, which represent a minor portion of the camera's active area. Notably, the total amount of light directed to the specific $k$ detection ports of the camera constitutes only
\begin{equation}
    \eta_{\text{ports}} = \frac{\sum_{j=1}^k\left|\alpha^{(x,y)}_j\right|^2}{\eta\alpha^2} \approx 8\%
\end{equation}
of the overall light impinging on the camera. This efficiency could be increased by using larger detection ports; however, once their size exceeds that of a speckle grain, they would also collect uncontrolled light carrying no information about Alice's encoding.

\subsection{Quantum \gls{BPM} and classical benchmark with \gls{VS}}
We are now ready to analyze the performance of our experiment and benchmark it against the best-known classical protocol for the \gls{VS} problem. 
For each partition of the input ports, we performed the \gls{BPM} protocol over $l=500$ instances of the variables $\bm{x} = (x_1, \dots, x_l)$ and $\bm{y} = (y_1, \dots, y_l)$. In particular, we fine-tuned the beam intensity and the exposure time of our \gls{EMCCD} camera to ensure the error probability estimated from the $l$ instances of the protocol was below a targeted error threshold.
  \begin{figure}[htb!]
         \centering
\includegraphics[width=0.8\textwidth]{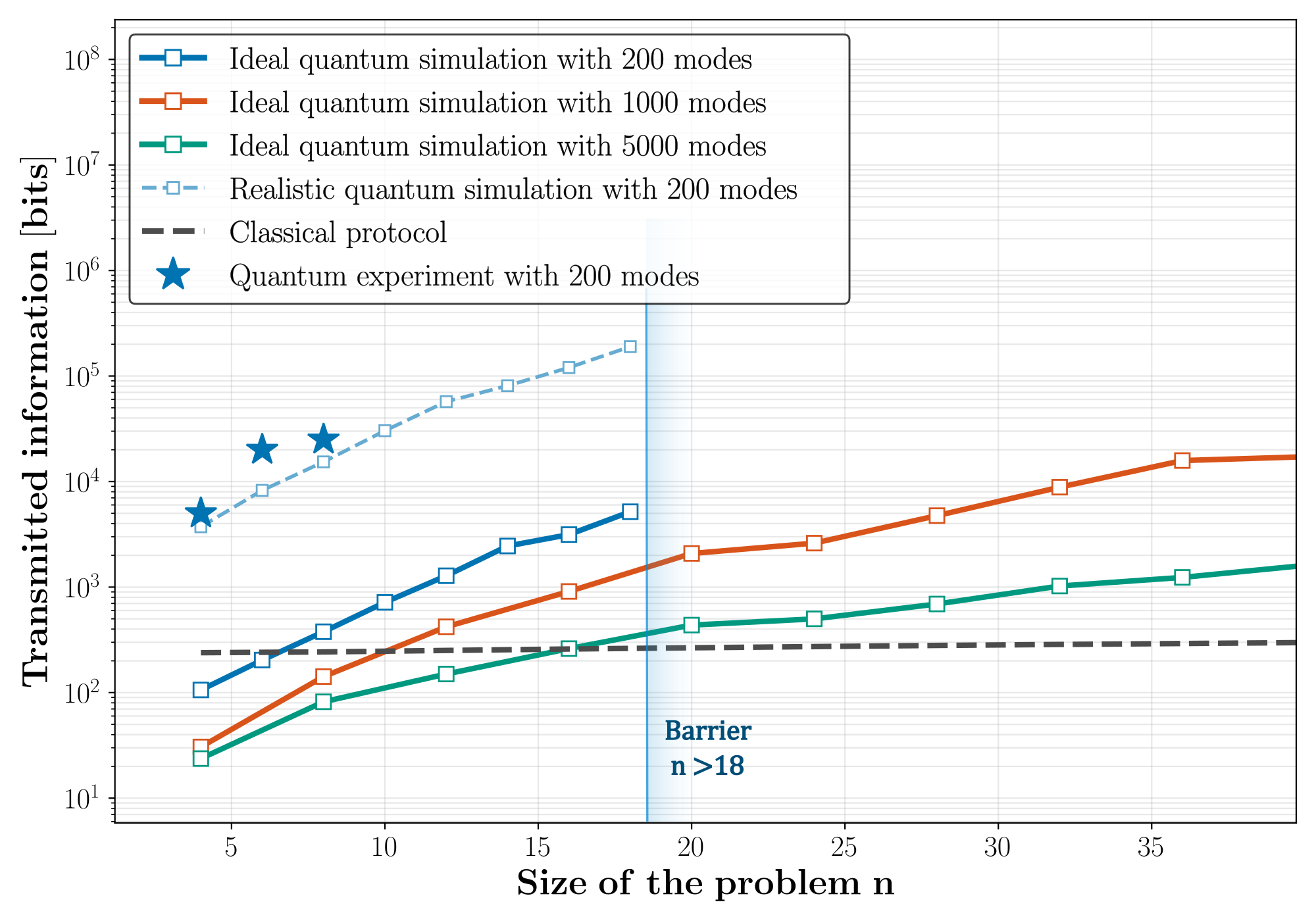}
\caption[Plot of the transmitted information vs.\ the input size $n$ for solving the $\beta$PM problem.]{Plot of the transmitted information vs.\ the input size $n$ for solving the \gls{BPM} problem within error probability $\epsilon = 0.2$. It includes quantum simulations for \gls{MMF}s with varying mode numbers ($N_{\text{modes}} = 200, 1000, 5000$), both in an idealized scenario with no end-to-end losses and no camera noise (solid line) and in a more realistic scenario including the measured setup losses and camera noise extrapolated from the datasheet (dashed line).
 For each mode count, the plots are color-coded (blue, red, and green, respectively). The barrier indicates the largest input size $n$ for which the target error can still be achieved; beyond this point, increasing the transmitted optical power no longer suffices to reduce the error below the required threshold. The graph also includes the transmitted bits required by the best-known classical protocol for the VS problem, estimated using the closed-form sufficient condition derived from Theorem \ref{theo:VS_communication} in Supplementary Corollary 1.}
\label{fig:exp_scaling}
 \end{figure}
 
We then measured the number of transmitted photons $|\alpha|^2$, and hence evaluated the transmitted information $|\alpha|^2 \log(n)$, by injecting the output of the photon source into an avalanche photodiode.
As shown in Figure \ref{fig:exp_scaling}, experimentally, we required $(5\cdot 10^3$, $2\cdot 10^4$, $2.5\cdot 10^4)$ bits of transmitted information to solve the \gls{BPM} problem with an error rate of $\epsilon = 0.2$ for $n=(4,6,8)$ ports, respectively.

We also compared the results to a set of numerical simulations for MMFs with $N_{\text{modes}}= (200, 1000, 5000)$ physical modes. Given a fiber with $N_{\text{modes}}$ physical modes and an optical configuration with $n$ inputs, each input port of the \gls{SLM} can be interpreted as effectively addressing approximately $N_{\text{modes}}/n$ physical modes. In the simulations, instead of directly measuring each single-port transmission matrix $\textbf{T}^{(p)}$, we therefore derived them by partitioning the full transmission matrix $\textbf{T}^{(\text{full})}$---measured with the entire \gls{SLM} treated as a single port---into $n$ submatrices. Specifically, we considered
$
\textbf{T}^{(\text{full})}
= \left[
\textbf{T}^{(1)}
\dots
\textbf{T}^{(n)}
\right].
$
As detailed in the Supplementary Material, the simulations were performed at pixel resolution, using the complex transmission-matrix coefficients for all $128\times256$ camera pixels in order to properly account for noise at the pixel level.

Finally, to simulate transmission matrices with a higher number of modes, we copied the initial transmission matrix and reshuffled the modes, thereby increasing the effective number of physical modes. In particular, we performed two types of simulations: (i) an ideal simulation, which sets a theoretical benchmark for the protocol, and (ii) a more realistic simulation that includes the detection noise characteristics of our camera, as detailed in the Supplementary Material. It is important to stress that the ideal simulation does not rely on a perfectly engineered linear transformation. In both cases, the simulations are performed using experimentally measured transmission matrices of the MMF, so that the mode-mixing properties of the actual device are retained. The ideal case differs from the realistic one in that it neglects end-to-end losses ($\eta = 1$) and assumes a noise-free camera.

Although our experiment is still several orders of magnitude away from the classical benchmark for the \gls{VS} problem, the ideal simulations demonstrate its potential to achieve a quantum advantage. Moreover, they highlight again how beneficial it would be to increase the number of effective modes of the \gls{MMF}.

It is however crucial to understand the gap between the ideal simulations and the experimental results. As shown in Fig.~\ref{fig:exp_scaling}, the realistic simulations are in good agreement with the measured data, indicating that the dominant limitations of the present implementation are well captured by the simulation including the measured losses and the camera noise.
In particular, this suggests that the current performance is largely constrained by the detection stage, rather than by a fundamental limitation of the wavefront-shaping approach itself. Significant improvements should therefore be achievable by optimizing the camera operation and, more generally, by increasing the signal-to-noise ratio of the detection stage.

Finally, one may wonder why the experiment is only implemented for the first three input sizes, while the realistic simulation for $N_{\mathrm{modes}}=200$ predicts successful operation up to the barrier at $n \approx 18$. The key point is that, although the realistic simulation accounts for the detection noise of the camera, it still assumes ideal wavefront control over the measured transmission matrix, namely perfect amplitude-and-phase shaping of the input optical field.
 In the experiment, however, the quality of the implemented transformation is further constrained by the finite number of SLM pixels that are effectively used, as well as by the fact that only a fraction of the total SLM surface lies within the active illuminated area (see Fig.~1 of the Supplementary Material). These limitations reduce the fidelity of the programmed operator and shift the experimental performance away from the simulated barrier. This gap could in principle be mitigated by increasing the number of effective SLM pixels and by optimizing the active area used for modulation.


\subsection{\gls{BPM} communication channel}
To further explore the trade-off between mean photon number and error rate, for different encoding size, we considered a  communication scenario where Alice and Bob use $l = 6 \cdot 10^3$ instances of the \gls{BPM} protocol to share a binary image of a fingerprint.  One can imagine that Alice and Bob have agreed in advance on a set of inputs $\bm{y} = (y_1,\dots, y_l)$. Alice can then simply choose accordingly each of her inputs $\bm{x} = (x_1,\dots, x_l)$ in order for her to transmit the fingerprint. Specifically the $j$-th binary pixel of the targeted image has the value $a_j = \beta \text{PM}(x_j,y_j)$.
 \begin{figure}[htb!]
         \centering
\includegraphics[width=0.6\textwidth]{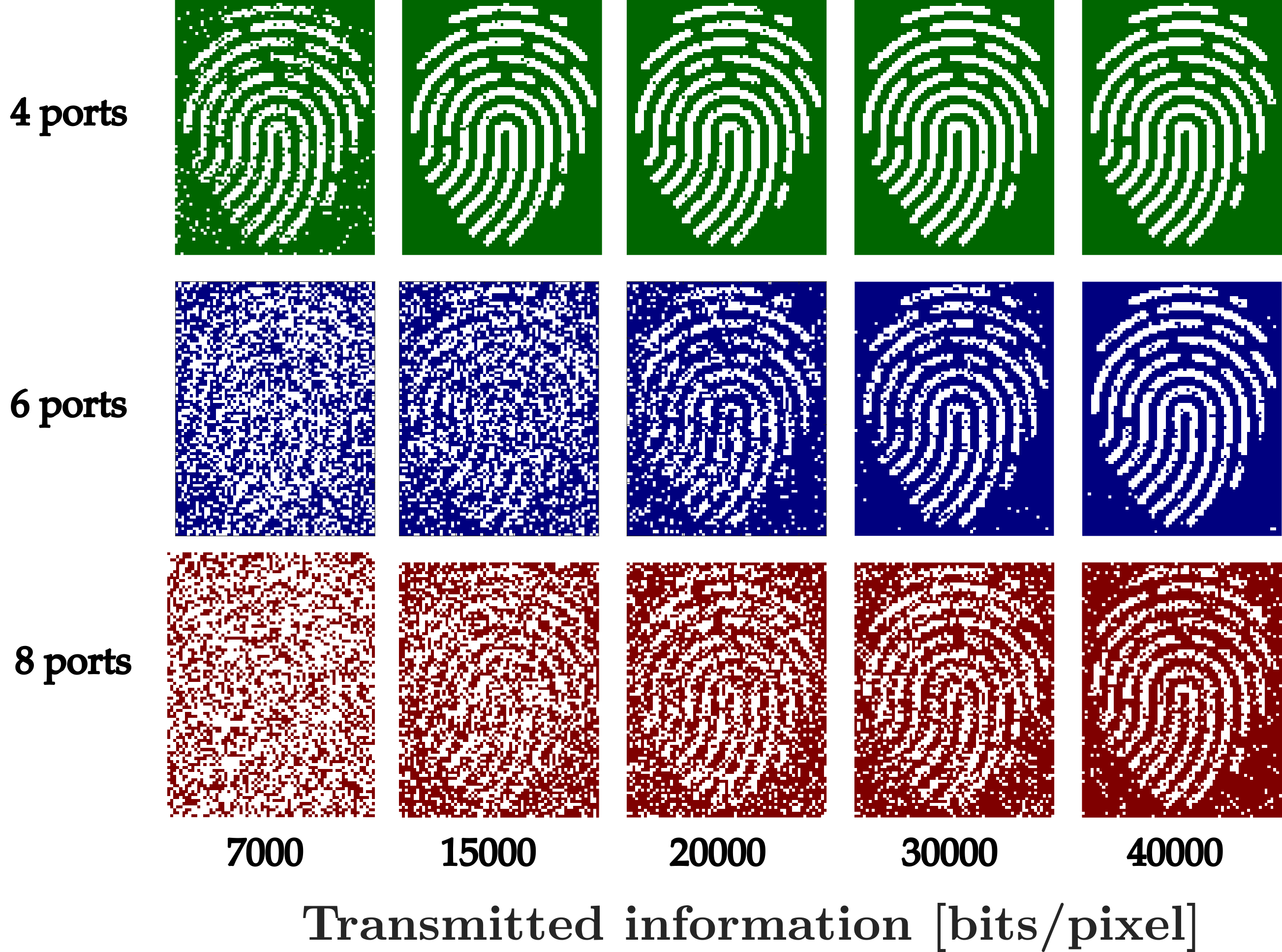}
\caption[Experimental transmission of a binary classical fingerprint image with $6\cdot 10^3$ pixels solving the $\beta$PM quantum protocol.]{Experimental transmission of a binary classical fingerprint image with $6\cdot 10^3$ pixels solving the \gls{BPM} quantum protocol.}
\label{fig:exp_fingerprint}
 \end{figure}

The experimental results in Figure \ref{fig:exp_fingerprint} clearly demonstrate that an image with $4$ ports becomes distinctly visible with as few as $7\cdot10^3$ bits of transmitted information per pixel. When the number of ports increases to $6$ and $8$, the quantity of transmitted information required for a clear image does rise significantly, yet the transmission of well-defined images remains achievable. For example, with $4\cdot 10^4$ bits transmitted per pixel, we can successfully transmit an image with $8$ ports while maintaining a pixel error rate of $\epsilon = 0.12$. See Figure \ref{fig:exp_fingerprint_plot}  for the full plot of the pixel error rates achieved for each image with respect to the number of transmitted information per pixel.

\begin{figure}[htb!]
         \centering
\includegraphics[width=0.65\textwidth]{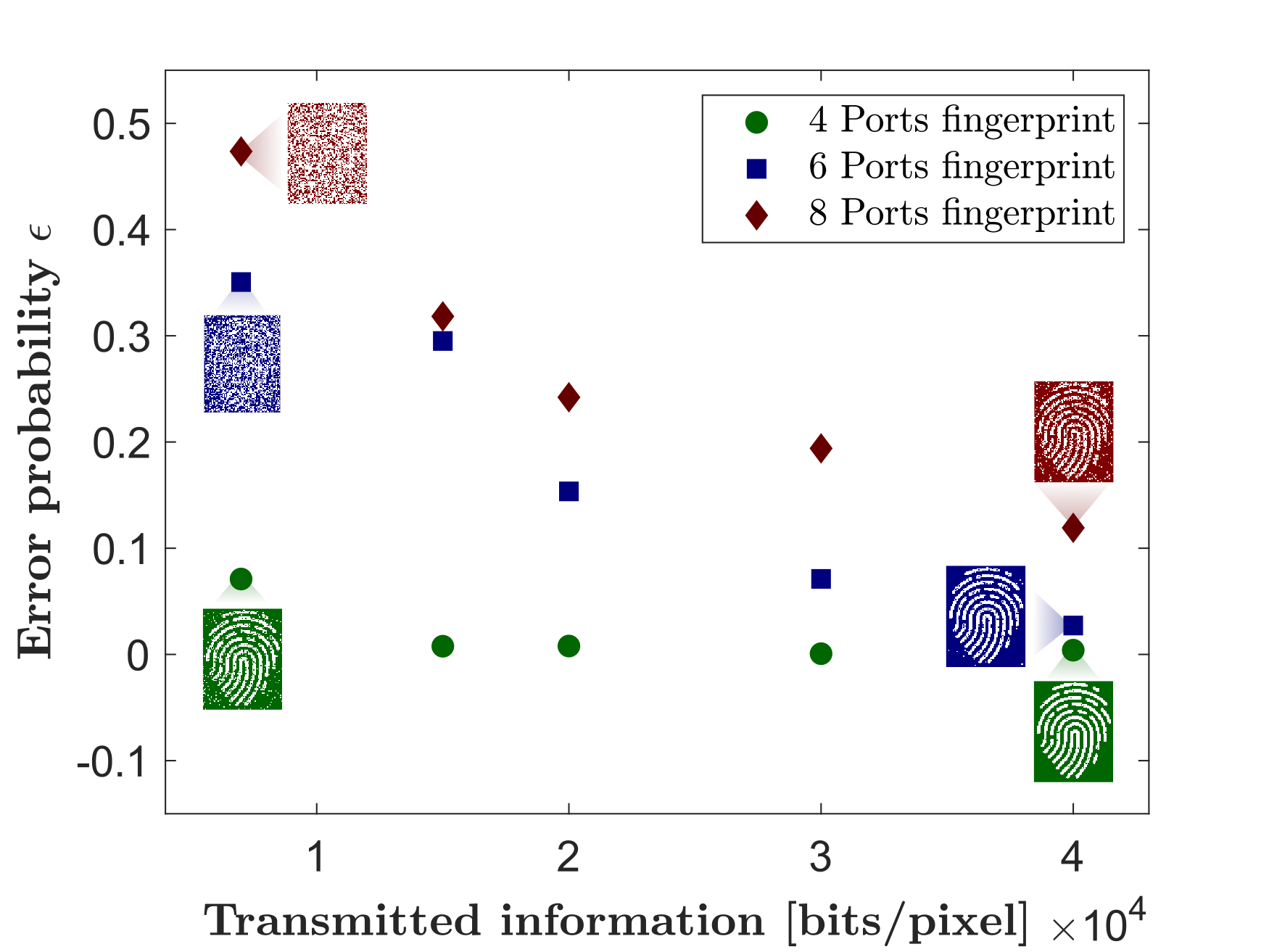}
\caption[Plot of error probabilities of the experimental transmission of a binary classical fingerprint solving the $\beta$PM problem.]{Plot of error probabilities of the experimental transmission of a binary classical fingerprint solving the \gls{BPM} problem.}
\label{fig:exp_fingerprint_plot}
 \end{figure}

 \newpage
\section{Discussion}
The investigation of novel quantum communication complexity problems with better separations stands as a vibrant and dynamic area of interest. In this context, having a sound and flexible platform for assessing quantum communication protocols and clearly showing the benefits of quantum technology is extremely valuable for advancing the field.
We have introduced one such optical platform leveraging spatial degrees of freedom, specifically designed for addressing quantum communication problems. Our strategy is based on the intricate mode mixing within a \gls{MMF}, coupled with the precise shaping of the incoming light wavefront.
This method’s key strength is its flexibility:  it provides a reconfigurable detection architecture with a very high level of programmability for one-way quantum communication complexity problems. This stands in stark contrast with previous experiments on communication complexity problems \cite{xuExperimentalQuantumFingerprinting2015a, kumarExperimentalDemonstrationQuantum2019}, where the optical setup was explicitly designed to perform a restricted set of linear operations. 

In particular, our reconfigurable detection platform is well suited to problems requiring the implementation of a large number of quantum linear operators. A representative example is the \gls{VS} problem, whose best-known classical protocol is much more demanding in terms of communication cost than the corresponding benchmark for the \gls{BPM} problem. Here, we observe that increasing the number of physical modes and reducing technical noise could move the platform toward a regime where practical one-way quantum implementations may outperform the best-known classical bounds. The number of physical modes could be increased, for instance, by engineering \gls{MMF}s with larger core diameters or higher numerical apertures.

 Finally,  we believe that by leveraging various multiplexing techniques, such as spectral, timing, and polarization, such an experimental platform can enable the solution of quantum communication problems in higher dimensions.  Another possibility to dramatically increase the size of the Hilbert space under manipulation is to use, as input, instead of single photons,  Fock states containing strictly more than one photon, offering another exciting area for further exploration.

\section{Methods}
 
 \subsection{Scalability and flexibility of the optical network}
The scalability and programmability of this general approach were studied in \cite{leedumrongwatthanakunProgrammableLinearQuantum2020} under the straightforward \gls{RMT} model, where the elements of the TMs are drawn from i.i.d.\ complex Gaussian random variables. Given a fiber with $N_{\text{modes}}$ physical modes and an optical configuration with $n$ input and $k$ outputs ports, the asymptotic scaling of the fidelity in this model  was found to be 
 \begin{equation}
         \mathcal{F}(\bm{\mathcal{L}}_{\text{obs}}, \bm{\mathcal{L}}) = 1- \mathcal{O}\Big( \sqrt{\frac{nk}{N_{\text{modes}}}}\Big)\;,
         \label{eq:fidelity_scaling}
 \end{equation}
defined as $\mathcal{F}(\bm{\mathcal{L}}_{\text{obs}}, \bm{\mathcal{L}}) \coloneqq 1 -\frac{\absnorm{\bm{\mathcal{L}}_{\text{obs}} - \bm{\mathcal{L}}}}{nk}$,
where $\absnorm{A} \coloneqq  \sum_{i=1}^k \sum_{j=1}^n |A_{ij}|$ is the $l_1$-vector norm.
There is therefore a critical relationship between the dimensions of linear operators and the number of physical modes $N_{\text{modes}}$. Specifically, to increase the dimensionality of square operators  ($k=n$) linearly, a quadratic increase is required in the number of physical modes ($N_{\text{modes}} = \mathcal{O}(n^2)$). 

 Finally, we are interested to see how this flexibility specifically applies when trying to solve \gls{VS} and \gls{BPM} problems, with a model which is closer to our experimental apparatus. Considering for the sake of simplicity lossless single-photon  protocols, we choose  as a figure of merit to compare the two protocols the \emph{visibility} $V$ of the boolean detection system, defined as the probability of each photon going in the right region of the camera
\begin{equation}
    V = \braket{\psi_{out}| \Pi(a) | \psi_{out}},
\end{equation}
 with $a \in \{0,1\}$ being the correct answer. The operators  $\Pi(0)= \sum_{i\in H} \ket{i}\!\bra{i},$ and $\Pi(1)= \sum_{i\in V} \ket{i}\!\bra{i}$
 represent the POVM associated with the boolean response. 
 \begin{figure}[htb!]
         \centering
         \includegraphics[width=0.6\textwidth]{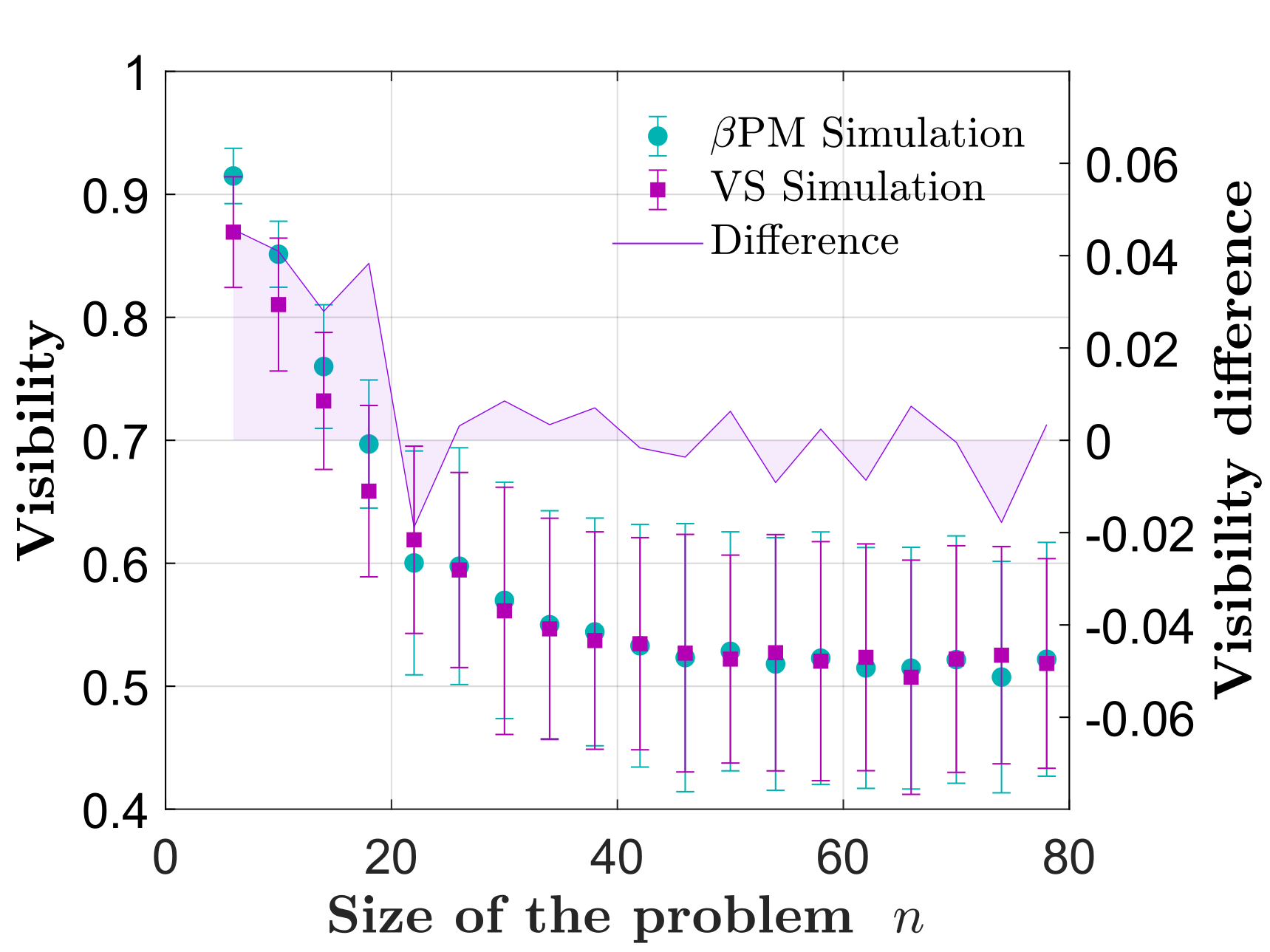}
\caption{Numerical simulation for the visibility of \gls{BPM} and \gls{VS} protocols. }
\label{fig:HM_vs_VSP_visibility}
 \end{figure}
Therefore, visibility can be interpreted as the probability of correctly guessing in a one-way quantum communication complexity problem, where a single photon is sent in a scenario without any loss. As displayed in Figure \ref{fig:HM_vs_VSP_visibility}, the visibility rapidly decreases with the size of the problem, which is in accordance with the scaling in Eq.\ \eqref{eq:fidelity_scaling}. Moreover, the visibility difference between the \gls{BPM} and \gls{VS} protocols is at most $4\%$, showcasing that the effectiveness of our detection method is influenced not by the specific linear operator used but solely by the size of the linear network. For the \gls{BPM} protocol, we set $\beta = 1/2$ in order to avoid additional losses due to partial matchings and ensure a fair comparison with \gls{VS}; the same choice is used in the experiment.

\subsection{Calibration}
To operate the optical setup as a reprogrammable optical network, we  perform a calibration routine. The setup is sufficiently stable that this calibration is required only approximately once a week.
First, we need to define how to subdivide it into distinct input ports. While a division in square grids has already been proposed in \cite{matthesOpticalComplexMedia2019}, we decide to preserve the radial symmetry of the Gaussian input source by dividing the \gls{SLM} into different pizza-shaped slices, which allows to more easily equalize intensities. However, we need first to characterize the Gaussian beam coming from the superluminescent diode and how it couples to the \gls{MMF}. To do this, we performed a detailed scan of the SLM, mapping out how light from each pixel of the \gls{SLM} impinges on the \gls{MMF} (see Supplementary Material).

Once we have characterized the active zone of the \gls{SLM}, we can then divide it into $m$ different slice-shaped ports, where our goal is to make sure that each port is controlling the same amount of light  (see Supplementary Material). In particular, this is crucial when we want to simulate a quantum state coming from Alice, such as the one for the \gls{BPM} protocol, which is a balanced superposition where the classical information is encoded in the phase. 

Finally, we fully characterize our \gls{MMF} by independently measuring the transmission matrix $\textbf{T}^{(p)}$, linking the relevant input modes for each $p$-th input port to the targeted output modes \cite{defienneTwophotonQuantumWalk2016}. In our experiment,  we define the input modes as a series of focused spots arranged on an isometric grid at the input facet of the \gls{MMF}.  To ensure a complete characterization of the fiber we always consider $N_{\text{input}} = 600$  different spots, a number higher than the supported modes by the MMF.  Each input mode corresponds to a specific phase ramp displayed on the \gls{SLM}, while the output modes are simply $k$  macro pixels of the \gls{EMCCD} camera,  each of which is approximately the size of an average speckle grain.
\vspace{0.5cm}

\noindent {\large \textbf{Acknowledgments}}\\
F.M. and R.A. acknowledge the financial support from the European Union’s Horizon Europe research and innovation programme under the project QSNP (Grant No. 101114043) and the Digital Europe programme under the project FranceQCI (project ID 101091675). R.A. acknowledges the financial support of France 2030, to the QCommTestbed project of the PEPR program, reference 22-PETQ-0011.
H.D. acknowledges funding from the ERC Starting Grant (number SQIMIC-101039375) and the ANR (ANR-24-CE97-0001 and ANR-23-CE47-0014).
S.G. is senior member of the IUF (institut Universitaire de France).
The authors thank Michał D\k{a}browski for assistance with the experimental setup preparation and for useful discussions on the setup limitations.

\newpage

\putbib[Bibliography_QCC3]
\end{bibunit}

\begin{bibunit}[quantum]
\clearpage
\pdfbookmark[0]{Supplementary Material}{supplementary-material}
\begingroup
\centering
{\large\bfseries Supplementary Material: Reconfigurable Optical Network for Quantum Communication Complexity\par}
\vspace{1.0em}
{Francesco Mazzoncini,$^{1,2}$ Hugo Defienne,$^{3}$ Romain All\'eaume,$^{1}$ and Sylvain Gigan$^{4}$\par}
\vspace{0.8em}
{\small
$^{1}$T\'el\'ecom Paris-LTCI, Institut Polytechnique de Paris, 19 Place Marguerite Perey, 91120 Palaiseau, France\par
$^{2}$Orange Innovation, Orange Gardens, 44 Avenue de la R\'epublique, 92326 Ch\^atillon, France\par
$^{3}$Sorbonne Universit\'e, CNRS, Institut des NanoSciences de Paris, INSP, Paris, France\par
$^{4}$Laboratoire Kastler Brossel, ENS-Universit\'e PSL, CNRS, Sorbonne Universit\'e, Coll\`ege de France, 24 rue Lhomond, Paris, France\par
}
\endgroup
\vspace{2em}

\setcounter{section}{0}
\setcounter{subsection}{0}
\setcounter{subsubsection}{0}
\setcounter{equation}{0}
\setcounter{figure}{0}
\setcounter{table}{0}
\setcounter{theorem}{0}
\setcounter{corollary}{0}
\setcounter{lemma}{0}
\glsresetall

\makeatletter
\renewcommand*{\theHsection}{supp.\arabic{section}}
\renewcommand*{\theHsubsection}{supp.\arabic{section}.\arabic{subsection}}
\renewcommand*{\theHequation}{supp.\arabic{equation}}
\renewcommand*{\theHfigure}{supp.\arabic{figure}}
\renewcommand*{\theHtable}{supp.\arabic{table}}
\providecommand*{\theHtheorem}{}
\renewcommand*{\theHtheorem}{supp.\arabic{theorem}}
\providecommand*{\theHcorollary}{}
\renewcommand*{\theHcorollary}{supp.\arabic{corollary}}
\providecommand*{\theHlemma}{}
\renewcommand*{\theHlemma}{supp.\arabic{lemma}}

\def\hyper@natlinkstart#1{%
  \Hy@backout{#1}%
  \hyper@linkstart{cite}{cite.supp.#1}%
  \def\hyper@nat@current{#1}%
}
\def\hyper@natlinkbreak#1#2{%
  \hyper@linkend#1\hyper@linkstart{cite}{cite.supp.#2}%
}
\def\hyper@natanchorstart#1{%
  \Hy@raisedlink{\hyper@anchorstart{cite.supp.#1}}%
}
\makeatother
\section{Complete simulation of the setup}
\subsection{Full-pixel output description}
When building the linear operator as in Eq.\ (7) of the main text, one needs to characterize the complex values transmission matrices only for a fixed number of detection ports $k$, each consisting of multiple pixels. However, to compute a complete simulation of the setup, we took into account the contribution of the transmission matrix for each  of the $N_{\text{pixel}}=  32768$ pixels of our camera. Let's call this $N_{\text{pixel}} \times \frac{N_{\text{input}}}{n}$ matrix  for an input port $p$ $\textbf{T}^{(p)}_{\text{pixel}}$. 
For each of the $p_1, \dots, p_k$ detection modes,  we select a target pixel $i^\star_{p_i}$ (typically the central pixel) on which to focus the light.  We call $S^\star = \{i^\star_{p_1}, \dots , i^\star_{p_K}\}$.
Starting from the linear operator for the input port $p$, $\bm{\mathcal{L}}^{(p)}$, we construct a “masked” operator, $\tilde{\bm{\mathcal{L}}}^{(p)}$ , consisting of a   $N_{\text{pixel}}$ vector obtained by zeroing all the elements corresponding to pixels outside the central pixel of the detection modes

\begin{equation}
\left(\tilde{\bm{\mathcal{L}}}^{(p)}\right)_i = \begin{cases} \left(\bm{\mathcal{L}}^{(p)}\right)_i, & \text{if } i \in S^\star, \\ 0 & \text{if } i \notin S^\star. \end{cases} \end{equation}

\noindent
One can now write an equation similar to Eq.\ (7)
 of the main text to model the observed contribution of the simulated linear operator over all the pixels of the camera

 \begin{equation}
    \bm{\mathcal{L}}_{\text{pixel}}^{(p)} = \textbf{T}^{(p)}_{\text{pixel}}\textbf{T}^{(p)\dagger} _{\text{pixel}}\bm{\tilde{\mathcal{L}}}^{(p)}\;.
\end{equation}

\noindent
Accordingly, one can derive a pixel-equivalent of Eq.\ (8) of the main text

\begin{align}
       \ket{\alpha\widetilde{, x,}y}  &= \bigotimes_{i=1}^{N_{\text{pixel}}}  \left |\alpha \sqrt{\frac{\eta}{n}}\sum_{p=1}^n (-1)^{x_p}\left(\textbf{T}^{(p)}_{\text{pixel}}\textbf{T}^{(p)\dagger} _{\text{pixel}}\bm{\tilde{\mathcal{L}}}^{(p)}_{\beta \text{PM}}(y)\right)_i\right \rangle_i \nonumber \\
         &=  \bigotimes_{i=1}^{N_{\text{pixel}}}\left |\tilde{\alpha}^{(x,y)}_i\right \rangle_i\;,
         \label{supp:eq:finalestate_pixel}
\end{align}

\noindent
where we defined $\tilde{\alpha}^{(x,y)}_i\coloneqq \alpha \sqrt{\frac{\eta}{n}}\sum_{p=1}^n (-1)^{x_p}\left(\textbf{T}^{(p)}_{\text{pixel}}\textbf{T}^{(p)\dagger} _{\text{pixel}}\bm{\tilde{\mathcal{L}}}^{(p)}_{\beta \text{PM}}(y)\right)_i$ the new complex amplitudes for the  $i$-th pixel of the camera and Bob's linear operator $\bm{\mathcal{L}}_{\beta \text{PM}}(y)$ is the operator used to solve the $\beta$PM protocol, fully characterized in Section \ref{supp:section: bpm_protocol} of this Supplementary Material.

\subsection{Camera noise analysis}
In this section we  introduce a general detailed physical model for
EMCCD noise properties \cite{hirschStochasticModelElectron2013}. In Table \ref{supp:table:symbols} we present all the mathematical symbols used to describe the model.

Let $n^i_{pe}$  be the random variable representing the number of photoelectrons generated by the signal impinging the $i$-th pixel. Since the output state  $\ket{\alpha\widetilde{, x,}y}$,  is a  tensor product of $N_{\text{pixel}}$ coherent states, these variables follow Poisson distributions

\begin{align*}
 \text{P}(n^i_{pe}\,|\,\alpha, x, y)& = \text{Poisson}(n^i_{pe}\,|\,\tilde{\mu}_i)\\
 &\coloneqq e^{-\tilde{\mu}_i}\,\frac{\tilde{\mu}_i^{\,n^i_{pe}}}{\bigl(n^i_{pe}\bigr)!},
\end{align*}
\noindent
where $\tilde{\mu}_i\coloneqq |\tilde{\alpha}^{(x,y)}_i|^2$.  From now on, we will drop the $i$ index for simplicity, as long as it is clear from the context that we are referring to the $i$-th pixel.

\begin{table}[h!]
\centering
\begin{tabular}{|c|l|}
\hline
\multicolumn{2}{|l|}{\textbf{Model parameters}} \\
\hline
$\alpha$ & \parbox[t]{10cm}{Coherent state amplitude} \\
$x, y$ & \parbox[t]{10cm}{Communication Complexity inputs} \\
$c$ & \parbox[t]{10cm}{Spurious charge (dark current and clock induced charge) in electrons} \\
$g$ & \parbox[t]{10cm}{Gain of the electron multiply (EM) register} \\
$r$ & \parbox[t]{10cm}{Readout noise in electrons} \\
$f$ & \parbox[t]{10cm}{A/D factor in electrons per image value} \\
$t_{exp}$ & \parbox[t]{10cm}{Exposure time} \\
$p_{dark}$ & \parbox[t]{10cm}{Dark current rate} \\
$c_{CIC}$ & \parbox[t]{10cm}{Clock Induced Charge rate} \\
\hline
\multicolumn{2}{|l|}{\textbf{Quantities of the signal flow}} \\
\hline
$n_{pe}$ & \parbox[t]{10cm}{Number of photoelectrons} \\
$n_{ie}$ & \parbox[t]{10cm}{Number of input electrons of the EM register} \\
$n_{oe}$ & \parbox[t]{10cm}{Number of output electrons of the EM register} \\
$n_{ic}$ & \parbox[t]{10cm}{Number of image counts (pixel value of the digital image)} \\
\hline
\end{tabular}
\caption{Table of mathematical symbols.}
\label{supp:table:symbols}
\end{table}

When deploying an EMCCD camera for low-light measurements, additional sources of noise must be considered. Alongside the photoelectrons generated by incoming photons, spurious charge also contributes to the noise. This spurious charge follows a Poisson distribution with an emission rate $c$. Specifically, this rate can be divided into two main components: Clock Induced Charge (CIC) and dark current. CIC arises from the rapid transfer of charge through the camera's pixels during readout, while dark current is generated by thermal electrons within the detector. While the dark current is time-dependent,
the CIC only depends on the number of readout processes.  Consequently, the overall emission rate per pixel is given by $c = t_{exp} \cdot p_{dark} + c_{CIC}$, where $t_{exp}$ is the exposure time. Thus, the number of electrons entering the EM register, $n_{ie}$, is described by the sum of these two contributions

\begin{equation}
 \text{P}(n_{ie}\,|\,\alpha, x, y, c) =\text{Poisson}(n_{ie}\,|\, \mu + c).
\end{equation}
These charges are then sent through the electron-multiplication register, which amplifies the number of input electrons by a mean gain factor $g$. Following the stochastic model introduced in Ref.~\cite{hirschStochasticModelElectron2013}, the conditional distribution of the number of output electrons $n_{oe}$ is described by a Poisson--gamma mixture. Let $\lambda := \mu+c$. The distribution of the
number of output electrons is
\begin{equation}
P(n_{oe}\mid\alpha,x,y,c,g)
=
\begin{cases}
\displaystyle
\sum_{n_{ie}=1}^{n_{oe}}
\operatorname{Poisson}(n_{ie}\mid\lambda)
G(n_{oe}\mid n_{ie},g),
& n_{oe}>0,\\[2ex]
e^{-\lambda},
& n_{oe}=0,\\
0,
& n_{oe}<0,
\end{cases}
\end{equation}
where
\begin{equation}
G(n_{oe}\mid n_{ie},g)
=
\frac{n_{oe}^{\,n_{ie}-1}e^{-n_{oe}/g}}
     {g^{n_{ie}}\Gamma(n_{ie})},
\qquad n_{oe}\geq n_{ie}>0.
\end{equation}

Finally, the readout register converts the analogue signal into discrete image values using the analogue-to-digital proportionality factor $f$, which is the number of electrons per image value. The readout noise is modeled by a normal distribution with standard deviation $r$, $N(fn_{ic}| n_{oe}, r)$.  Thus, the probability of measuring the digital pixel value $n_{ic}$  can be written as

\begin{equation}
  \text{P}(n_{ic}\,|\,\alpha, x, y, c, g, r)   =  \sum_{n_{oe}=0}^{\infty} N(f n_{ic}\,|\, n_{oe}, r) \cdot\text{P}(n_{oe}\,|\,\alpha, x, y, c, g) .
\end{equation}
Finally, the protocol’s output bit is obtained by summing the pixel values over the active $H$ and $V$ regions of the camera and comparing the two totals.

\subsection{Modeling losses in the full simulation}
When running the full simulation, it is essential to specify how optical losses are modeled and where they are applied. This is particularly important for quantum communication complexity: overestimating the overall efficiency $\eta$ leads to a systematic underestimation of the total communication cost, roughly by a factor $\simeq 1/\eta$. Moreover, when interpreting communication cost—especially toward cryptographic uses \cite{mazzonciniHybridQuantumCryptography2025a}—one should distinguish losses affecting photons that actually carry information from losses incurred on photons that do not.

\begin{itemize}
    \item 
      Source-side (Alice) losses. In principle, losses occurring on Alice’s setup can be factored out of the communication cost, provided that the corresponding lost photons do not leak to an adversary. In standard cryptographic models, an eavesdropper is assumed to have access only to Alice’s output state, i.e. the state injected into the quantum channel, so internal attenuation before this point is not charged as communication. However, in our implementation Alice and Bob share a single SLM (see Figure 2 in the main text); as a result, the SLM’s attenuation takes place after Alice’s logical encoding and directly affects the state that would enter the channel. We therefore treat SLM‑induced attenuation as loss on information‑carrying photons and include it, as a factor  $\eta_{\text{SLM}}$,  in the end-to-end efficiency $\eta$.

\item Transmission (Alice $\rightarrow$ Bob) losses. In an ideal one-way quantum communication complexity protocol, Alice and Bob are physically separated; thus, the communication medium (free space, MMF, multicore fibers, etc.) and the distance introduce additional attenuation. By contrast, in our experiment both parties operate on the same SLM, so there is effectively no physical channel. Accordingly, we model zero path length and add no extra transmission loss (i.e., $\eta_{\text{channel}} = 1$).
\item  Bob's detection losses.  As explained in the ``Setup efficiency'' subsection of the main text, here the main contributions—beyond the SLM already discussed—are the  losses induced by the multimode fiber, with overall efficiency $\eta_{\text{MMF}}$, and the camera  quantum efficiency $\eta_{\text{det}}$.
\end{itemize}

Therefore, the end-to-end transmissivity used in Eq.~\eqref{supp:eq:finalestate_pixel} is
\begin{equation}
  \eta \;=\; \eta_{\text{SLM}} \cdot \eta_{\text{channel}} \cdot \eta_{\text{MMF}} \cdot \eta_{\text{det}}.
\end{equation}
These factors are parameters we can set in the simulation.
In contrast, the efficiency term $\eta_{\text{ports}}$—arising from the fact that only the pixels in the  $H$ and $V$ regions are used as detection modes—is \emph{not} a free parameter. It is determined by the transmission matrices and the chosen linear operator as 
\begin{equation}
  \eta_{\text{ports}}
  \;=\;
  \frac{\sum_{i \in H \cup V} \left|\tilde{\alpha}^{(x,y)}_i\right|^2}{\eta\,\alpha^2}.
\end{equation}

\subsection{Different scenarios considered}
This noise model was used to generate the simulations shown in Figure~4 of the main text. In particular, we considered two scenarios: an ideal case with no end-to-end losses and a noise-free camera, and a more realistic case using the noise parameters of our Andor iXon3 860, as reported in the datasheet and summarized in Table~\ref{supp:table:comparison}. It is important to stress that, in both scenarios, we used experimentally measured transmission matrices. Therefore, the ideal case does not correspond to a perfectly engineered linear transformation, but only to a simulation in which detection noise is neglected and the overall efficiency is set to $\eta = 1$. By contrast, the realistic case includes the measured losses of our setup, yielding an end-to-end efficiency of $\eta = 54\%$.

\begin{table}[h!]
\centering
\begin{tabular}{|c|c|c|}
\hline
\textbf{Parameter} & \textbf{Ideal Case} & \textbf{Realistic case} \\
\hline
\(\eta\) & $100\%$ & $54\%$ \\
\hline
\(r\) & 0 & 49 \(e^-\)  \\
\hline
\(p_{dark}\) & 0 & \(1 \cdot 10^{-3}\) e/(pixel\(\cdot\)sec) \\
\hline
\(c_{CIC}\) & 0 & \(5 \cdot 10^{-3}\) e/pixel  \\
\hline
\(g\) & No EM amplifier & 300 \\
\hline
\end{tabular}
\caption{Parameters used in the two simulation scenarios. In both cases, the simulations are performed using experimentally measured transmission matrices. The ideal case assumes no end-to-end losses and a noise-free camera, whereas the realistic case includes the measured losses of the setup and the camera noise parameters.}

\label{supp:table:comparison}
\end{table}

\section{$\beta$PM Quantum protocol}

\label{supp:section: bpm_protocol}

In the quantum protocol for the $\beta$PM problem, Alice prepares a coherent state encoding her input bit string $x$ as described in Eq.~(6) of the main text. Bob's decoding strategy is then to apply a $2\beta n \times n$ linear operator $\bm{\mathcal{L}}_{\beta \mathrm{PM}}(y)$---which is fully determined by his classical input $y = (M, \omega)$---to the incoming coherent state. 
 This linear operator can be decomposed into a sequence of permutations and pairwise interference operations
\begin{equation}
\bm{\mathcal{L}}_{\beta \mathrm{PM}}(y) = \mathbf{P}^{(\omega)} \mathbf{H} \mathbf{P}^{(M)}\;.
\label{supp:eq:unitary_BPM}
\end{equation}

Given the matching $M = (i_1, j_1), (i_2, j_2), \ldots, (i_{\beta n}, j_{\beta n})$, the permutation matrix $\mathbf{P}^{(M)}$ of size $2\beta n \times n$ is constructed as follows. For $k = 1, 2, \ldots, \beta n$,
\begin{equation}
    P^{(M)}_{2k - 1, i_k} = 1, \quad P^{(M)}_{2k, j_k} = 1,
\end{equation}
and all other entries of $\mathbf{P}^{(M)}$ are zero. This operation selects and reorders the $n$ input modes so that each pair corresponding to an edge in the matching $M$ is mapped to adjacent output modes, effectively filtering out modes not involved in the matching
\begin{align}
    \mathbf{P}^{(M)}\ket{\alpha, x} &= \bigotimes_{k=1}^{\beta n} 
    \left(\left|\frac{\alpha}{\sqrt{n}}(-1)^{x_{i_k}}\right\rangle_{2k-1}
    \otimes
    \left|\frac{\alpha}{\sqrt{n}}(-1)^{x_{j_k}}\right\rangle_{2k}\right)\;.
\end{align}
Next, the operator $\mathbf{H}$ implements pairwise interference between the modes associated with each edge, and is given by
\begin{equation}
\mathbf{H} = \mathbb{I}_{\beta n} \otimes H_2,
\end{equation}
where $H_2$ is the $2 \times 2$ Hadamard matrix acting on each pair, and $\mathbb{I}_{\beta n}$ is the identity on the set of edges. This step mixes the amplitudes of each pair, such that the resulting coherent state encodes the parity information of Alice's input bits for each edge

\begin{align}
      \mathbf{H} \mathbf{P}^{(M)} \ket{\alpha, x} &=
    \bigotimes_{k=1}^{\beta n}
    \left|
        \frac{\alpha}{\sqrt{2n}}
        \left(
            (-1)^{x_{i_k}} + (-1)^{x_{j_k}}
        \right)
    \right\rangle_{2k-1}
    \otimes
    \left|
        \frac{\alpha}{\sqrt{2n}}
        \left(
            (-1)^{x_{i_k}} - (-1)^{x_{j_k}}
        \right)
    \right\rangle_{2k}\,.
\end{align}

After these operations, the modes are reorganized according to Bob's auxiliary input $\omega$. Let $\tilde{\mathbf{P}}^{(\omega)}$ be a $2\beta n \times 2\beta n$ permutation matrix defined by the binary vector $\omega$, where for each pair $(2i-1, 2i)$ with $i = 1, 2, \ldots, \beta n$,
\begin{equation}
\tilde{P}^{(\omega)}_{jk} = 
\begin{cases} 
    1 & \text{if } (j, k) = (2i - 1 + \omega_i, 2i - 1) \text{ or } (j, k) = (2i - \omega_i, 2i), \\
    0 & \text{otherwise}.
\end{cases}
\end{equation}
This ensures that if $\omega_i = 1$, the $i$-th pair $(2i-1, 2i)$ is swapped, while if $\omega_i = 0$, the pair remains unchanged.

To facilitate the final detection, we introduce a permutation matrix $\mathbf{G}$ that reorganizes all the modes such that the modes with odd indices are mapped to the first $\beta n$ output modes (H region of the camera) and the even ones to the last $\beta n$ modes (V region of the camera)
\begin{equation}
G_{ij} = 
\begin{cases} 
    1 & \text{if }\, j = 2i - 1 \,\text{ for } i = 1, \ldots, \beta n, \\
    1 & \text{if }\, j = 2(i - \beta n) \, \text{ for } i =  \beta n + 1, \ldots, 2\beta n, \\
    0 & \text{otherwise}.
\end{cases}
\end{equation}
The final permutation matrix $\mathbf{P}^{(\omega)}$ is then obtained by combining the two operators
\begin{equation}
    \mathbf{P}^{(\omega)} = \mathbf{G} \tilde{\mathbf{P}}^{(\omega)}.
\end{equation}

After applying the full linear operator $\bm{\mathcal{L}}_{\beta \mathrm{PM}}(y)$, the output is a set of coherent states whose intensities in the two detection regions encode the answer to the $\beta$PM problem. The measurement is performed by summing the detected intensities in each region, and the protocol's output is determined by which region receives the higher total intensity.

\section{Input ports calibrations}
\label{supp:appendix:input_ports}

\subsection{Active zone of the SLM}
\label{supp:appendix:active_zone}
 Our approach begins with displaying a phase ramp on the \gls{SLM} and collecting the resulting speckle pattern. The key observation comes from noticing changes in the speckle pattern only when we vary the phase value of pixels that reflect light inside the \gls{MMF}. Therefore,  as illustrated in Figure \ref{supp:fig:activezone}, by adding a phase shift of $\pi$ to each macro pixel (each sized $4\times 4$) and analyzing the correlation $ C(X,Y)$ between the original speckle image (without any additional phase shift) and the modified image (with the phase shift at macro pixel$(x,y)$) we can effectively reconstruct the Gaussian beam's profile. 
  In particular, this technique allows us to accurately find the position of the beam's center, which is crucial to divide the \gls{SLM} into well-balanced slices.
  \begin{figure}[htb!]
         \centering
\includegraphics[width=0.95\textwidth]{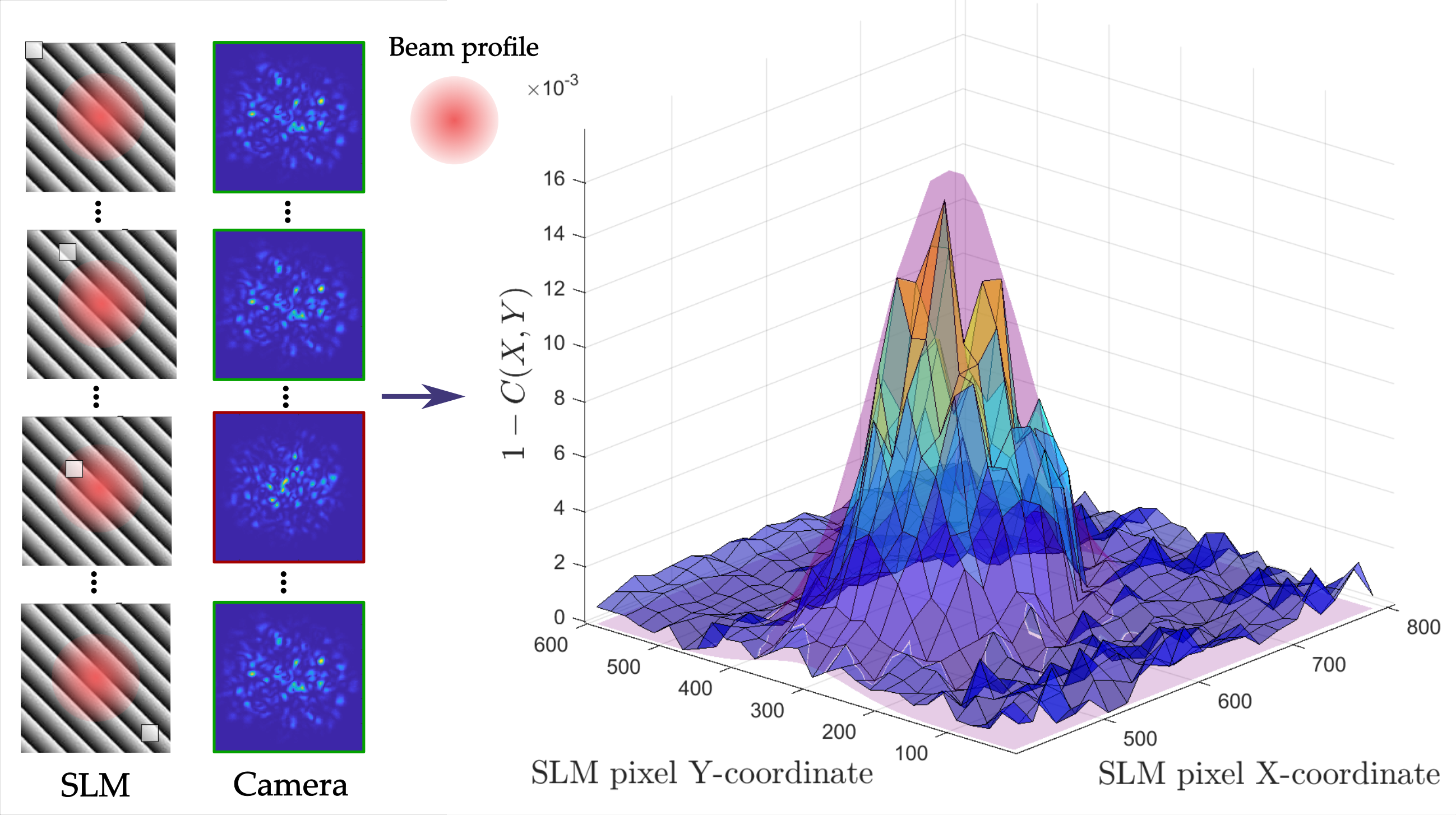}
\caption[Analysis of the active zone of the SLM.] {Analysis of the active zone of the \gls{SLM}. The reconstruction of the Gaussian beam's profile is achieved by applying a $\pi$ phase shift to each macro pixel, identified by coordinates $(x,y)$, and examining the correlation $C(X,Y)$ between the original and modified speckle images. }
\label{supp:fig:activezone}
 \end{figure}

 \subsection{Dividing the SLM in slices}
\label{supp:appendix:slm_slicing}
Our technique for dividing \gls{SLM} into $m$ well-balanced ports, illustrated in Figure \ref{supp:fig:slices_scanning}, works as follows. Initially, we display a uniform phase ramp across the entire \gls{SLM} and record the total light intensity, $ \text{I}^{tot} $. Next, we set the SLM to a completely black state (phase $ 0 $ across all areas), ensuring no light reaches the \gls{MMF} as it aligns with the first diffraction order. We then apply the phase ramp to a specific slice of the SLM, defined by a starting angle $ \alpha_1^{in} = 0 $ and a variable ending angle $ \alpha_1^{out} $, initially set to $ 0 $. This angle is gradually increased until the light intensity collected from this section, $ \text{I}_1^{tot} $, reaches $ \text{I}^{tot}/m $. Subsequent ports are scanned in the same way starting from the ending angle of the previous port $ \alpha_p^{in} = \alpha_{p-1}^{out} $. The final port $ m $ is uniquely defined by the start angle $ \alpha_m^{in} = \alpha_{m-1}^{out} $ and ends at $ \alpha_m^{out} = 2\pi $, completing the circle.




  \begin{figure}[htb!]
         \centering
\includegraphics[width=0.7\textwidth]{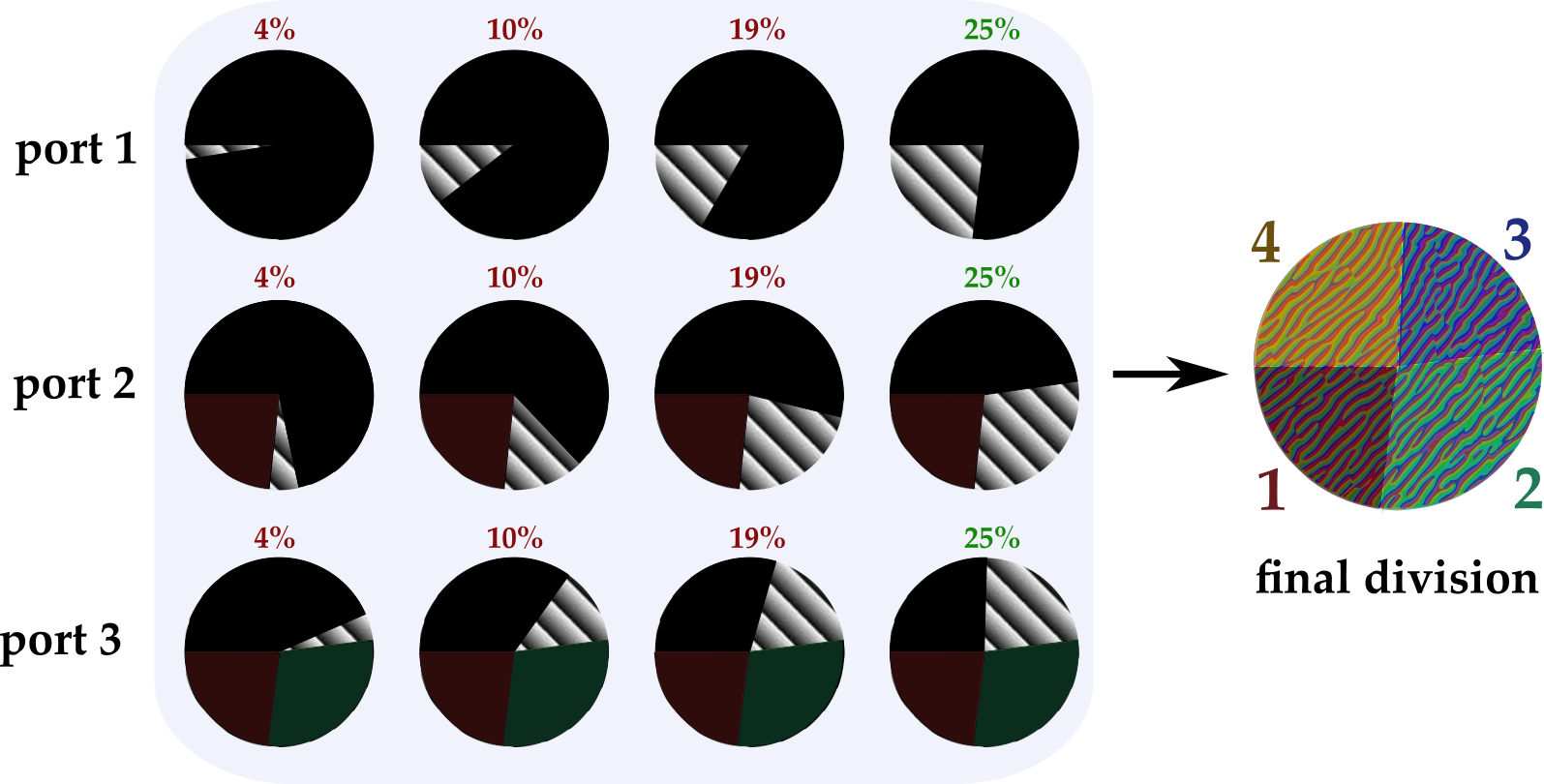}
\caption[Illustration of the division of the SLM into $m=4$ different ports.]{Illustration of the division of the \gls{SLM} into $m=4$ different ports. Due to the fact that the \gls{SLM} active profile is not perfectly Gaussian, the angles of the four ports are not precisely 90° each.}
\label{supp:fig:slices_scanning}
 \end{figure}

  \section{Vector in a Subspace best-known protocol}\label{supp:appendix:VS}
We analyze the best-known classical protocol $\piVS$, sketched in \cite{razExponentialSeparationQuantum1999}.  \\[0.2cm]
\textbf{Classical protocol $\bm{\piVS}$.}
Let $d\in \mathbb N$ be the communication budget, and set
$
k:=2^d$.
Using public randomness, Alice and Bob sample $k$ independent Gaussian vectors
$
z^1,\dots,z^k \sim \mathcal N\!\left(0,\frac1n I_n\right)$,  meaning that each $z^j$ is a centered Gaussian vector in $\mathbb R^n$ with covariance matrix $\frac1n I_n$.  Given the input vector $x\in S^{n-1}$, Alice computes
$
\hat j\in\arg\max_{1\le j\le k}\langle z^j,x\rangle
$
and sends the index $\hat j$ to Bob using $d$ bits. Since the vectors $z^1,\dots,z^k$ are publicly known, Bob can reconstruct $z^{\hat j}$ from the received index. He outputs $0$ if
$
\|\mathrm{Proj}_H z^{\hat j}\|^2 \ge \|\mathrm{Proj}_{H^\perp} z^{\hat j}\|^2,
$
and outputs $1$ otherwise.

The following theorem provides an explicit upper bound on the error probability of the classical protocol \(\piVS\). In particular, it implies that a communication cost of order \(\mathcal{O}(\sqrt n)\) is sufficient to achieve constant error probability.

\begin{theorem}[VS best-known classical protocol]
\label{supp:theo:VS_communication}
Let $n\ge 4$ be even and let $d\ge 1$ be an integer. An explicit one-way public-coin protocol $\piVS$ exists with communication cost $CC(\piVS)=d$ which solves the $n$-dimensional \gls{VS} problem with error probability at most
\begin{equation}
\label{supp:eq:epsilon_VS}
\epsVS(d,n)
=
\inf_{\begin{array}{c}
0<t<\sqrt{d/\pi},\,
u>0
\end{array}}
\left[
e^{-t^2/2}
+
2e^{-u^2/2}
+
e^{-x_{t,u}}
\right],
\end{equation}
where $x_{t,u}$ is defined as follows. If the set of nonnegative numbers $x$ such that
\begin{equation}
\label{supp:eq:def_xtu}
\frac{n-1+2\sqrt{(n-1)x}+2x}{n}
\le
\frac{\left(\sqrt{d/\pi}-t\right)^2}{\frac{n}{n-1}+\frac{2un}{\sqrt{n-1}}}
\end{equation}
is nonempty, then $x_{t,u}$ is its largest element. Otherwise, we set $x_{t,u}=0$.
\end{theorem}

\begin{proof}

We analyze the protocol $\piVS$ defined above. By symmetry, it is enough to consider the case $x\in H$, since the case $x\in H^\perp$ is obtained by exchanging $H$ and $H^\perp$.

We start with two standard results on maxima of Gaussian random variables.

\begin{lemma}[Borell-TIS inequality \cite{adlerGaussianInequalities2007}]
\label{supp:lemma:Borel-Tis}
Let $Y=\max_{1\le i\le m}X_i$, where $X_i\sim\mathcal{N}(0,\sigma^2)$ are i.i.d.\ Gaussian random variables. Then for every $t>0$
\begin{equation}
\label{supp:eq:Borel-Tis}
P\left(Y-\mathbb{E}[Y]<-t\right)\le e^{-t^2/(2\sigma^2)}.
\end{equation}
\end{lemma}

\begin{lemma}[\cite{kamathBoundsExpectationMaximum}]
\label{supp:lemma:Emax}
Let $Y=\max_{1\le i\le m}X_i$, where $X_i\sim\mathcal{N}(0,\sigma^2)$ are i.i.d.\ Gaussian random variables. Then
\begin{equation}
\label{supp:eq:Emax}
\mathbb{E}[Y]\ge \frac{1}{\sqrt{\pi\ln 2}}\sigma\sqrt{\ln m}.
\end{equation}
\end{lemma}

We will also use the following concentration bound for a random projection on the sphere.

\begin{lemma}[Levy's lemma  \cite{aubrunAliceBobMeet2017}]
\label{supp:lemma:Levy}
Let $\nu\in S^{N-1}$ be a unit random vector distributed uniformly, with $N>2$. Let $F\subseteq \mathbb{R}^N$ be a subspace of dimension $l$. Then for each $u>0$
\begin{equation}
\label{supp:eq:Levy}
P\left(
\left|
|\mathrm{Proj}_F \nu|^2-\frac{l}{N}
\right|
\ge \frac{u}{\sqrt N}
\right)
\le 2e^{-u^2/2},
\end{equation}
where $|\cdot|$ is the Euclidean norm and $\mathrm{Proj}_F \nu$ is defined as the projection of $\nu$ onto the subspace $F$, formally given by $\mathrm{Proj}_F \nu \coloneqq \arg\min_{v \in F} |\nu - v|$.
\end{lemma}

Finally, we will use the following standard upper-tail inequality for chi-square random variables.

\begin{lemma}[Chi-square upper-tail inequality {\cite[Lemma~1, eq.~(4.3)]{laurentAdaptiveEstimationQuadratic2000}}]
\label{supp:lemma:chi_square_tail}
Let $U\sim\chi^2(D)$. Then for every $x>0$,
\begin{equation}
\label{supp:eq:chi_square_tail}
P\left(U-D\ge 2\sqrt{Dx}+2x\right)\le e^{-x}.
\end{equation}
\end{lemma}

For each $j\in{1,\dots,k}$, with $k=2^d$, we decompose the Gaussian vector $z^j$ into its component along the direction $x$ and its orthogonal component. More precisely, we define
$$
X_j:=\langle z^j,x\rangle
\qquad\text{and}\qquad
W_j:=z^j-\langle z^j,x\rangle x.
$$
Since $x$ is a unit vector, $X_j$ is the scalar coefficient of the projection of $z^j$ onto the one-dimensional space spanned by $x$, while $W_j$ belongs to the orthogonal subspace
$
x^\perp:={q\in\mathbb R^n:\langle q,x\rangle=0}.
$
Therefore,
$$
z^j=xX_j+W_j.
$$

Now recall that
$
z^j\sim\mathcal N\left(0,\frac1n I_n\right)
$. Equivalently, the coordinates of $z^j$ in any orthonormal basis are independent Gaussian random variables with variance $1/n$.
It follows that the scalar projection $X_j=\langle z^j,x\rangle$ is a one-dimensional centered Gaussian random variable with variance
$$
\mathrm{Var}(X_j)=\frac1n|x|^2=\frac1n,
$$
and hence
$
X_j\sim\mathcal N\left(0,\frac1n\right)$.
Similarly, $W_j$ is a centered Gaussian vector supported on the subspace $x^\perp$. Its covariance is the restriction of $\frac1n I_n$ to $x^\perp$, so in an orthonormal basis of $x^\perp$ its coordinates are independent Gaussian random variables with variance $1/n$. In this sense,
$
W_j\sim\mathcal N\left(0,\frac1n I_{x^\perp}\right)$.
Here, $I_{x^\perp}$ denotes the identity operator on the $(n-1)$-dimensional Euclidean space $x^\perp$.

Finally, since $z^j$ is an isotropic Gaussian vector, its components along orthogonal subspaces are independent. Therefore, the scalar variable $X_j$ and the orthogonal component $W_j$ are independent.

Let
$
y:=\max{X_1,\dots,X_k}=\langle z^{\hat{j}},x\rangle.
$
Applying Lemma~\ref{supp:lemma:Emax} with $\sigma=n^{-1/2}$ and $m=2^d$, we obtain
$$
\mathbb{E}[y]\ge \sqrt{\frac{d}{n\pi}}.
$$
Hence, by Lemma~\ref{supp:lemma:Borel-Tis}, for every $t>0$,
\begin{equation}
\label{supp:eq:boundy}
P\left(
y\ge \frac{1}{\sqrt n}\left(\sqrt{\frac d\pi}-t\right)
\right)
\ge 1-e^{-t^2/2}.
\end{equation}
From now on, we restrict to $0<t<\sqrt{d/\pi}$, so that the lower bound in Eq.~\eqref{supp:eq:boundy} is nonnegative. Hence, whenever the event in Eq.~\eqref{supp:eq:boundy} holds, we also have
\begin{equation}
\label{supp:eq:boundy_square}
y^2 \ge \frac{1}{n}\left(\sqrt{\frac d\pi}-t\right)^2.
\end{equation}

Next, the choice of $\hat{j}$ only depends on the scalar variables $X_1,\dots,X_k$. Since the family $(W_1,\dots,W_k)$ is independent of $(X_1,\dots,X_k)$, the selected orthogonal component $W_{\hat{j}}$ still has distribution $\mathcal{N}(0,n^{-1}I_{x^\perp})$. Therefore we may write
$$
z^{\hat{j}}=xy+r\nu,
$$
where
$
r:=|W_{\hat{j}}|,
$
and $\nu$ is uniformly distributed on the unit sphere of $x^\perp$.

Let
$
q:=|\mathrm{Proj}_{H^\perp}\nu|^2.
$
Since $x\in H$, we have $H^\perp\subseteq x^\perp$. Moreover, since $n$ is even, $\dim(H^\perp)=n/2$. Hence, inside the $(n-1)$-dimensional space $x^\perp$, the mean value of $q$ is $\frac{n/2}{n-1}$. Applying Lemma~\ref{supp:lemma:Levy} with $N=n-1$ and $l=n/2$, we obtain that for every $u>0$
\begin{equation}
\label{supp:eq:boundq}
P\left(
q\le \frac{n}{2(n-1)}+\frac{u}{\sqrt{n-1}}
\right)
\ge 1-2e^{-u^2/2}.
\end{equation}

Now,
$$
|\mathrm{Proj}_{H^\perp} z^{\hat{j}}|^2=r^2q,
\qquad
|\mathrm{Proj}_{H} z^{\hat{j}}|^2=y^2+r^2(1-q).
$$
Therefore Bob answers correctly whenever
$$
y^2+r^2(1-q)>r^2q,
$$
that is,
\begin{equation}
\label{supp:eq:correct_condition}
y^2>r^2(2q-1).
\end{equation}
Assume now that the events in Eqs.~\eqref{supp:eq:boundy_square} and \eqref{supp:eq:boundq} both hold. Then
$$
y^2 \ge \frac{1}{n}\left(\sqrt{\frac d\pi}-t\right)^2,
\qquad
2q-1 \le \frac{1}{n-1}+\frac{2u}{\sqrt{n-1}}.
$$
Since $r^2\ge 0$, it follows that
$$
r^2(2q-1)\le r^2\left(\frac{1}{n-1}+\frac{2u}{\sqrt{n-1}}\right).
$$
Therefore, a sufficient condition for Eq.~\eqref{supp:eq:correct_condition} is
$$
\frac{1}{n}\left(\sqrt{\frac d\pi}-t\right)^2>r^2\left(\frac{1}{n-1}+\frac{2u}{\sqrt{n-1}}\right).
$$
Equivalently,
\begin{align}
r^2
&<
\frac{
\frac{1}{n}\left(\sqrt{\frac d\pi}-t\right)^2
}{
\frac{1}{n-1}+\frac{2u}{\sqrt{n-1}}
}
\\
&=
\frac{
\left(\sqrt{d/\pi}-t\right)^2
}{
\frac{n}{n-1}+\frac{2un}{\sqrt{n-1}}
}
=: T_{t,u}.
\end{align}
\noindent
Hence, Bob certainly answers correctly whenever $r^2<T_{t,u}$.
Finally, since $r^2\sim \frac1n\chi^2(n-1)$, the random variable
$
U:=nr^2
$
has distribution $\chi^2(n-1)$. Hence, by Lemma~\ref{supp:lemma:chi_square_tail}, for every $x>0$,
$$
P\left(
nr^2-(n-1)\ge 2\sqrt{(n-1)x}+2x
\right)\le e^{-x},
$$
or equivalently,
$$
P\left(
r^2\ge \frac{n-1+2\sqrt{(n-1)x}+2x}{n}
\right)\le e^{-x}.
$$

If the set defining $x_{t,u}$ is empty, then $x_{t,u}=0$ and the trivial bound
$$
P(r^2\ge T_{t,u})\le 1=e^{-x_{t,u}}
$$
holds. Otherwise, by definition of $x_{t,u}$,
$$
\frac{n-1+2\sqrt{(n-1)x_{t,u}}+2x_{t,u}}{n}\le T_{t,u},
$$
and therefore
$$
P(r^2\ge T_{t,u})
\le
P\left(
r^2\ge \frac{n-1+2\sqrt{(n-1)x_{t,u}}+2x_{t,u}}{n}
\right)
\le e^{-x_{t,u}}.
$$
Combining this with Eqs.~\eqref{supp:eq:boundy} and \eqref{supp:eq:boundq} and using the union bound, we obtain
$$
\epsVS(d,n)
\le
e^{-t^2/2}
+
2e^{-u^2/2}
+
e^{-x_{t,u}}.
$$
Since this bound holds for every $0<t<\sqrt{d/\pi}$ and every $u>0$, taking the infimum over $t$ and $u$ concludes the proof.

\end{proof}

The previous theorem gives an optimized upper bound on the error probability, but its numerical evaluation still requires an optimization over the auxiliary parameters $t$ and $u$. The following corollary gives a simpler closed-form sufficient condition on the communication cost, obtained by fixing explicit values of these parameters.

\begin{corollary}[Closed-form sufficient communication cost]
\label{supp:cor:VS_closed_form}
Let $n\ge 4$ be even and let $0<\epsilon<1$. Define
$$
t_\epsilon:=\sqrt{2\ln\left(\frac{3}{\epsilon}\right)},
\qquad
u_\epsilon:=\sqrt{2\ln\left(\frac{6}{\epsilon}\right)},
\qquad
x_\epsilon:=\ln\left(\frac{3}{\epsilon}\right).
$$
Moreover, define
$$
A_{\epsilon,n}
:=
\frac{n}{n-1}
+
\frac{2u_\epsilon n}{\sqrt{n-1}},
\quad \textit{and} \quad
B_{\epsilon,n}
:=
\frac{
n-1+2\sqrt{(n-1)x_\epsilon}+2x_\epsilon
}{n}.
$$
Then the classical protocol $\piVS$ solves the $n$-dimensional \gls{VS} problem with error probability at most $\epsilon$ whenever
\begin{equation}
\label{supp:eq:closed_form_d}
d
\ge
\pi
\left(
t_\epsilon
+
\sqrt{A_{\epsilon,n}B_{\epsilon,n}}
\right)^2.
\end{equation}
\end{corollary}

\begin{proof}
We apply Theorem~\ref{supp:theo:VS_communication} with the explicit choice
$$
t=t_\epsilon,
\qquad
u=u_\epsilon.
$$
By definition of $t_\epsilon$ and $u_\epsilon$, we have
$$
e^{-t_\epsilon^2/2}=\frac{\epsilon}{3},
\qquad
2e^{-u_\epsilon^2/2}=\frac{\epsilon}{3}.
$$
It remains to control the third term appearing in the bound of Theorem~\ref{supp:theo:VS_communication}. We will show that, under the condition in Eq.~\eqref{supp:eq:closed_form_d}, the value $x_\epsilon$ is admissible in Eq.~\eqref{supp:eq:def_xtu} for the choice $t=t_\epsilon$ and $u=u_\epsilon$.

Indeed, for $t=t_\epsilon$ and $u=u_\epsilon$, the right-hand side of Eq.~\eqref{supp:eq:def_xtu} is
$$
\frac{
\left(\sqrt{d/\pi}-t_\epsilon\right)^2
}{
\frac{n}{n-1}+\frac{2u_\epsilon n}{\sqrt{n-1}}
}
=
\frac{
\left(\sqrt{d/\pi}-t_\epsilon\right)^2
}{
A_{\epsilon,n}
}.
$$
On the other hand, the left-hand side of Eq.~\eqref{supp:eq:def_xtu} evaluated at $x=x_\epsilon$ is exactly
$$
\frac{
n-1+2\sqrt{(n-1)x_\epsilon}+2x_\epsilon
}{n}
=B_{\epsilon,n}.
$$
Therefore, $x_\epsilon$ is admissible whenever
$$
B_{\epsilon,n}
\le
\frac{
\left(\sqrt{d/\pi}-t_\epsilon\right)^2
}{
A_{\epsilon,n}
}.
$$
This condition is implied by
$$
\sqrt{d/\pi}
\ge
t_\epsilon+\sqrt{A_{\epsilon,n}B_{\epsilon,n}},
$$
which is equivalent to Eq.~\eqref{supp:eq:closed_form_d}.

Hence, under Eq.~\eqref{supp:eq:closed_form_d}, the number $x_\epsilon$ belongs to the admissible set in Eq.~\eqref{supp:eq:def_xtu}. Since $x_{t_\epsilon,u_\epsilon}$ is defined in Theorem~\ref{supp:theo:VS_communication} as the largest element of this admissible set, we have
$$
x_{t_\epsilon,u_\epsilon}\ge x_\epsilon.
$$
Consequently,
$$
e^{-x_{t_\epsilon,u_\epsilon}}
\le
e^{-x_\epsilon}
=
\frac{\epsilon}{3}.
$$
Using Theorem~\ref{supp:theo:VS_communication} with $t=t_\epsilon$ and $u=u_\epsilon$, we therefore obtain
$$
\epsVS(d,n)
\le
e^{-t_\epsilon^2/2}
+
2e^{-u_\epsilon^2/2}
+
e^{-x_{t_\epsilon,u_\epsilon}}
\le
\frac{\epsilon}{3}
+
\frac{\epsilon}{3}
+
\frac{\epsilon}{3}
=
\epsilon.
$$
This concludes the proof.
\end{proof}

\begingroup
\renewcommand{\label}[1]{}
\putbib[Bibliography_QCC3]
\endgroup
\end{bibunit}

\end{document}